\documentclass[aps,rmp,twocolumn]{revtex4-1}

\usepackage{amsmath,amsfonts,amssymb,amsthm, bbm}

\usepackage{enumerate}  
\usepackage[pdftex]{graphicx} 
\usepackage{rotating}  

\newcommand{\bsa}{\mathcal B_s(\mathcal H)}
\newcommand{\tsa}{\mathcal T_s(\mathcal H)}
\newcommand{\qeff}{\mathcal E(\mathcal H)}

\newcommand{\cstate}{\mathcal S(\Omega,\Sigma)}
\newcommand{\qstate}{\mathcal S(\mathcal H)}
\newcommand{\meas}{(\Omega,\Sigma)}
\newcommand{\smeas}{\mathcal M_{\mathbb R}(\Omega,\Sigma)}
\newcommand{\measf}{\mathcal F_{\mathbb R}(\Omega,\Sigma)}
\newcommand{\bsadual}{\bsa^*}

\newcommand\herm{\mathbb H(\mathcal H)}
\newcommand\unit{\mathbb U(\mathcal H)}
\newcommand\den{\mathbb D(\mathcal H)}
\newcommand\eff{\mathbb E(\mathcal H)}
\newcommand\proj{\mathbb P(\mathcal H)}
\newcommand\povm{\mathrm{POVM}(\mathcal H)}
\newcommand\pvm{\mathrm{PVM}(\mathcal H)}

\newcommand\range{\mathrm{Ran}}

\newcommand\Tr{\mathrm{Tr}}

\newcommand\spec{\mathrm{spec}}
\newcommand\id{\mathbbm 1}

\newcommand\freq{\frac{2\pi i}{d}}
\newcommand\w[1]{e^{-\freq #1}}

\newcommand\ip[2]{\langle#1 , #2\rangle}
\newcommand\op[2]{|#1\rangle\!\langle#2|}
\newcommand\abs[1]{\left|#1\right|}
\newcommand\av[1]{\left\langle #1 \right\rangle}
\newcommand\norm[1]{\|#1\|}
\newcommand\conj[1]{\overline{#1}}
\newcommand\ket[1]{\left|#1\right\rangle}

\newtheorem{definition}{Definition}

\newtheorem{theorem}{Theorem}
\newtheorem{lemma}{Lemma}

\usepackage{hyperref} 

\begin{document}

\title{Quasi-probability representations of quantum theory with applications to quantum information science}
\author{Christopher Ferrie}
\affiliation{
Institute for Quantum Computing and Department of Applied Mathematics,
University of Waterloo,
Waterloo, Ontario, Canada, N2L 3G1}

\begin{abstract}
This article comprises a review of both the quasi-probability representations of infinite-dimensional quantum theory (including the Wigner function) and the more recently defined quasi-probability representations of finite-dimensional quantum theory.  We focus on both the characteristics and applications of these representations with an emphasis toward quantum information theory.  We discuss the recently proposed unification of the set of possible quasi-probability representations via frame theory and then discuss the practical relevance of negativity in such representations as a criteria for quantumness.
\end{abstract}

\date{\today}

\maketitle


\tableofcontents

\section{Introduction\label{S:Introduction}}
The quasi-probability representations of quantum theory evolved from the \emph{phase space} representations. Phase space is a natural concept in classical theory since it is equivalent to the state
space. The idea of formulating quantum theory in phase space dates back to the early days of
quantum theory when the Wigner function was introduced \cite{Wigner1932On}.  The term quasi-probability refers to
the fact that the function is not a true probability density as it takes on negative values for some
quantum states.

The Wigner phase space formulation of quantum theory is
equivalent to the usual abstract formalism of quantum theory in the same sense that Heisenberg's
matrix mechanics and Schrodinger's wave mechanics are equivalent to the abstract formalism.
The Wigner function representation is not the only quasi-probabilistic formulation of quantum theory.  However, in
most, if not all, of the familiar quasi-probability representations the kinematic
or ontic space of the representation is presumed to be
the usual canonical phase space of classical physics. In
the broader context of attempting to represent quantum
mechanics as a classical probability theory, the classical state space need not necessarily correspond to the phase
space of some classical canonical variables.  But even these more general
approaches are not applicable to describing quantum systems with a finite set of distinguishable
states.

Recently, analogues of the Wigner function for finite-dimensional
quantum systems have been proposed.  There are many adequate reviews of infinite dimensional phase representations \footnote{See, for example, \cite{Hillery1984Distribution,Lee1995Theory}.}, which we complement by reviewing the recently defined quasi-probability representations of finite-dimensional quantum states.  Moreover, we review a unifying formalism proposed in references \cite{Ferrie2008Frame,Ferrie2009Framed,Ferrie2010Necessity} which encompasses all previously constructed quasi-probability representations and provides a means of defining the class of all possible quasi-probability representations.

While entanglement is the most notable criteria for quantumness, it can only be defined for composite systems.  A non-negative
quasi-probability function is a true probability distribution, prompting some authors to suggest
that the presence of negativity in this function is a defining signature of non-classicality \emph{for a single system}.
Thus, a central concept in studies of the quantum-classical contrast in the
quasi-probability formalisms of quantum theory is the appearance of negativity.   However, it should be understood that any particular representation is non-unique and to some extent arbitrary since a state with negativity in one representation is positive in another.  The concept of entanglement arose from the study of \emph{nonlocality}.  This, and \emph{contextuality}, are perhaps the most famous notions of quantumness.  We will discuss the relationship between these and negativity in the final section.  Nonlocality and contextuality are terms often used as shorthand for the phrases ``there does not exist a \emph{local} hidden variable model for quantum theory'' and ``there does not exist a \emph{non-contextual} hidden variable model for quantum theory''.  Negativity is a similar concept.  Moreover, since the assumptions which go into the meta-concept of \emph{hidden variable model} are common to each notion, there is a sense in which it is convenient to call them ``equivalent'' \cite{Spekkens2008Negativity}.

The outline of the paper is as follows.  In section \ref{sec:probtoquasi} we see how quasi-probability arises quite naturally when trying to apply standard probabilistic constraints on quantum theory.  In section \ref{S:Infinite} the quasi-probability representations on infinite dimensional Hilbert space are reviewed.  Those defined for finite dimensional Hilbert spaces are review in section \ref{S:Finite}.  Both are shown to be part of a unified formalism utilizing \emph{frames} in section \ref{S:Frame}. Negativity as a criteria for quantumness and its relation to contextuality and nonlocality are discussed in section \ref{S:negativity}.  The final section summarizes what we have learned.

\section{From probability to quasi-probability\label{sec:probtoquasi}}

\subsection{Operational theories, ontological models and non-contextuality\label{sec:operational}}

An \emph{operational theory} \cite{Hardy2001Quantum,Spekkens2005Contextuality,Harrigan2007Ontological} is an attempt to mathematically model a real physical experiment.  The concepts in the theory are \emph{preparations}, \emph{systems}, \emph{measurements} and \emph{outcomes}.  A preparation $P$ is a proposition specifying how a real physical system has come to be the object of experimental investigation.  We will reason about a set of mutually exclusive preparations $\mathcal P$.  The system is assumed to then be measured according to some measurement procedure which produces an unambiguous answer called the outcome.  The measurement procedure is specified by a proposition $M$ belonging to a set of mutually exclusive possibilities $\mathcal M$.  The outcome is specified by a proposition $k$ which is assumed to belong to a set of mutually exclusive and exhaustive possibilities $K$.  This means that, in any given run of the experiment, exactly one of $K$ is true.

For example, $P$ could be the setting of the knob on a ``black box'' which sends objects to another ``black box'' with knob setting $M$ which outputs some sensory cue (e.g. an audible ``click'' or flash of light) labeled $k$ as each system passes through.  Note that $P$ and $M$ need not be statements about devices in a laboratory \cite{Schack2003Quantum}; $P$ could be a statement about a photon arriving from the sun, for example.

For a fixed preparation and measurement the outcome may not be deterministic.  The role of the theory is then to describe the probability of each each outcome conditional on the the various combinations of preparations and measurements.  That is, the theory should tabulate the numbers $\Pr(k|P\wedge M)$ for all $k,P$ and $M$.  For fixed $P$ and $M$, the mutually exclusive and exhaustive property of $K$ implies $\sum_k \Pr(k|P\wedge M)=1$.  An operational theory is then a specification $(\mathcal P,\mathcal M, K, \{\Pr(k|P\wedge M)\})$.

Quantum theory is an example of an operational theory where each preparation $P\in\mathcal P$ is associated with a density operator $\rho_P\in\den$ via the mapping $\rho:\mathcal P\to\den$ (for a explanation of the notation used in this article, see the Appendix).  In general, this mapping is not required to be injective or surjective; different preparations may lead to the same density operator and there may not exist a preparation which leads to every density operator.  Similarly, each measurement $M\in\mathcal M$ and outcome $k\in K$ is associated with an effect $E_{M,k}\in\eff$ via the mapping $E:\mathcal M\times K\to \eff$.  Again this mapping need not be injective or surjective.  Quantum theory prescribes the probability distributions $\Pr(k|P\wedge M)=\Tr(\rho_P E_{M,k})$ which is called the Born rule.  Since $K$ is a set of mutually exclusive and exhaustive possibilities, for fixed $M$, $\{E_{M,k}\}\in\povm$.

Notice that an operational theory only specifies the probabilities given the preparations in the set $\mathcal P$.  Suppose, for example, we are told a coin is tossed which determines which of two preparations procedures are implemented in a laboratory experiment.  The operational theory does not offer the probability of a given outcome conditional on this information.  However, the laws of probability provide us with the tools necessary to derive the desired probabilities.  In general, the task is, given any disjunction $\bigvee_i P_i$ of propositions (the set of which is denoted $D(\mathcal P)$) from the set of preparations $\mathcal P$, determine the probability distribution $\Pr(k|\bigvee_i P_i\wedge M)$.  Since the set of preparations is mutually exclusive, the laws of probability dictate
\begin{equation}\label{probability_prep_disjunt}
\Pr(k|\textstyle\bigvee_i P_i\wedge M)=\displaystyle\sum_i \Pr(k|P_i\wedge M)\Pr(P_i).
\end{equation}
Similarly, for any disjunction $\bigvee_j M_j$ of propositions from the set of measurements $\mathcal M$
\begin{equation}\label{probability_measure_disjunt}
\Pr(k|P\wedge \textstyle\bigvee_j M_j)=\displaystyle\sum_j \Pr(k|P\wedge M_j)\Pr(M_j).
\end{equation}
Putting Equations \eqref{probability_prep_disjunt} and \eqref{probability_measure_disjunt} together yields
\begin{equation}\label{logical_extension}
\Pr(k|\textstyle\bigvee_i P_i\wedge \textstyle\bigvee_j M_j)=\displaystyle\sum_{ij} \Pr(k|P_i\wedge M_j)\Pr(P_i)\Pr(M_j).
\end{equation}

Consider the example of quantum theory again.  Let $\bigvee_i P_i$ be an arbitrary disjunction of preparations from $\mathcal P$.  Then for fixed $k$ and $M$
\begin{align*}
\Pr(k|\textstyle\bigvee_i P_i\wedge M)&=\sum_i \Pr(k|P_i\wedge M)\Pr(P_i)\\
&=\sum_j \Tr(\rho_{P_i} E_{M,k})\Pr(P_i)\\
&= \Tr\left(\sum_i\Pr(P_i)\rho_{P_i} E_{M,k}\right)
\end{align*}
suggesting we define the function $\hat\rho:D(\mathcal P)\to\den$
\begin{equation}\label{def:prepdisjunct}
\hat\rho: \textstyle\bigvee_i P_i\mapsto \sum_i\Pr(P_i)\rho_{P_i}.
\end{equation}
Similarly, let us define a function $\hat E: D(\mathcal M)\times K\to \eff$ as
\begin{equation}\label{def:measdisjunct}
\hat E: (\textstyle\bigvee_j M_j,k) \mapsto \sum_j\Pr(M_j) E_{M_j,k}.
\end{equation}
Now, using $\hat \rho$ and $\hat E$ and defining $P_D:=\bigvee_i P_i$ and $M_D=\bigvee_j M_j$ we have
\begin{equation}\label{logical_Born}
\Pr(k|\textstyle\bigvee_i P_i\wedge \textstyle\bigvee_j M_j)=\Tr(\hat\rho_{P_D}\hat E_{M_D,k}).
\end{equation}
This works for the probability distributions of any operational theory and hence the set $(D(\mathcal P), D(\mathcal M), K, \mathcal H, \hat\rho,\hat E)$ is an operational theory itself.  Unless specified otherwise we will simply assume this analysis has been done when given an operational theory.  This is equivalent to assuming the $\rho$ and $E$ of an operational theory $(\mathcal P, \mathcal M, K, \mathcal H, \rho, E)$ have ranges which are convex sets.

An \emph{ontological model} is an attempt at interpreting an operational theory as an effectively epistemic theory of a deeper model describing the \emph{real state of affairs} of the system.  An ontological model posits a set $\Lambda$ of mutually exclusive and exhaustive \emph{ontic states} $\lambda$.  Each preparation $P$ is assumed to output the system in a particular ontic state $\lambda$.  However, the experimental arrangement may not allow one to have knowledge of which state was prepared.  This ignorance is quantified via a conditional probability distribution $\Pr(\lambda|P)$.  When this probability is viewed as a mapping $\mathcal P\to \textrm{Prob}(\Lambda)$, from preparations to the set of probability distributions on $\Lambda$, it is neither injective or surjective.  That is, in general, different preparation may give the same probability distribution of $\Lambda$ and, certainly, not all probability distributions are realized.  A measurement $M$ may not deterministically give an outcome $k$ which reveals the state $\lambda$ of the system.  Each measurement can then only be associated with an conditional probability distribution $\Pr(k|M\wedge\lambda)\in \textrm{Prob}(M)$.  Again, when this probability is viewed as a mapping $\mathcal M\times K \to \textrm{Prob}(M)$ it is not assumed to be injective or surjective.  An ontological model also requires that $\Lambda$ be such that knowledge of $\lambda$ renders knowledge of $P$ irrelevant.  In the language of probability,
\begin{itemize}
\item[CI] $k$ is conditionally independent of $P$ given $\lambda$.
\end{itemize}
Summarizing, an ontological model is a set $(\mathcal P, \mathcal M, K, \Lambda, \{\Pr(\lambda|P)\}, \{\Pr(k|M\wedge\lambda)\})$ such that CI holds.  As a consequence of CI, the law of total probability states
\begin{equation}\label{OM_total_prob}
\Pr(k|P\wedge M)=\sum_\lambda \Pr(k|M\wedge \lambda)\Pr(\lambda|P).
\end{equation}

Loosely speaking, \emph{non-contextuality} encodes the property that operationally equivalent procedures are represented equivalently in the ontological model \cite{Spekkens2005Contextuality}.  Two preparations are indistinguishable operationally if no measurement exists for which the probability of any outcome is different between the two.  A ontological model is non-contextual (with respect to it preparations) if the probability distributions over the ontic space the operationally equivalent preparations produce are equal.  Similarly, measurements can be operationally equivalent and the ontological model can be non-contextual with respect to them.  A mathematically concise definition of non-contextuality is as follows.

\begin{definition}\label{def:non-contextuality}
Let $(\mathcal P, \mathcal M, K, \Lambda, \{\Pr(\lambda|P)\}, \{\Pr(k|M\wedge \lambda)\})$ define an ontological model.
\begin{enumerate}[(a)]
\item Let $P,P'\in\mathcal P$.  The ontological model is \emph{preparation non-contextual} if for all $k\in K$, $M\in\mathcal M$ and $\lambda\in\Lambda$
\[\Pr(k|P\wedge M)=\Pr(k|P'\wedge M)\Rightarrow \Pr(\lambda|P)=\Pr(\lambda|P').\]
\item Let $M,M'\in\mathcal M$.  The ontological model is \emph{measurement non-contextual} if for all $k\in K$, $P\in\mathcal P$ and $\lambda\in\Lambda$
\[\Pr(k|P\wedge M)=\Pr(k|P\wedge M')\Rightarrow \Pr(k|M\wedge \lambda)=\Pr(k|M'\wedge \lambda).\]
\item The model is simply called a \emph{non-contextual ontological model} if it is both preparation and measurement non-contextual.
\end{enumerate}
\end{definition}

\subsection{Quantum theory, contextuality and quasi-probability\label{sec:quantum_non-contextuality}}

Recall that quantum theory as an operational model is defined via the mappings $\rho$ and $E$.  As an example, consider only those preparations associated with pure states.  These pure states are determined by a mapping $\rho=:\psi:\mathcal P \to \proj$.  The model of Beltrametti and Bugajski \cite{Beltrametti1995A} posits the ontic space $\Lambda=\proj$ and a sharp probability distribution
\begin{equation}\label{BB_prob}
\Pr(\lambda|P)d\lambda=\delta(\lambda-\psi_P)d\lambda
\end{equation}
for each preparation $P$.  This model suggests that the quantum state provides a complete specification of reality.  Recall each measurement procedure is associated with a POVM via the mapping $E:\mathcal M\times K\to \eff$.  Each measurement procedure $M$ implies a conditional probability
\begin{equation}\label{BB_meas}
\Pr(k|M\wedge \lambda)=\ip{\lambda}{E_{M,k}\lambda}.
\end{equation}
To show this is an ontological model, it remains only to verify that Equation \eqref{OM_total_prob} is satisfied.  It follows that
\begin{align*}
\Pr(k|P\wedge M)&=\int_\Lambda \Pr(k|M\wedge \lambda)\Pr(\lambda|P)d\lambda \\
&=\int_\Lambda \ip{\lambda}{E_{M,k}\lambda}\delta(\lambda-\psi_P)d\lambda \\
&=\ip{\psi_P}{E_{M,k}\psi_P},
\end{align*}
which is the Born rule for pure states.  The Beltrametti-Bugajski model is an ontological model for pure state quantum theory which is preparation non-contextual.  However, the range of the mapping $\psi$ in this case is not convex.  As we will now see, if we looked at all possible logical disjunctions of preparations in this model, so that range of the new mapping $\rho$ is $\den$, quantum theory admits no non-contextual model.  The most famous example of an ontological model for pure state quantum theory is the de Broglie-Bohm (or Bohmian) theory, which suffers the same flaw of admitting no non-contextual extension to the full (i.e. including density operators) quantum theory.  This model is discussed later in section \ref{S:context locality}.

\begin{lemma}\label{convexity lemma prep}
Suppose the convex range of $E$ contains a basis (for $\herm$) and the ontological model $(\mathcal P, \mathcal M, K, \Lambda, \{\Pr(\lambda|P)\}, \{\Pr(k|M\wedge \lambda)\})$ is preparation non-contextual.  Then, there exists a affine mapping (with domain the range of $\rho$) $\mu:\range(\rho)\to \textrm{Prob}(\Lambda)$ satisfying $\mu(\lambda|\rho_{P})=\Pr(\lambda|P)$.
\end{lemma}
\begin{proof}
For equivalent disjunctions of preparations
\begin{equation}\label{equivalent prep to equivalent den}
\Tr(\rho_{P}  E_{M,k})=\Pr(k|P\wedge M)=\Pr(k|P'\wedge M)= \Tr(\rho_{P'}  E_{M,k}).
\end{equation}
Since the range of $ E$ is a basis it spans $\herm$ and therefore $\rho_P=\rho_{P'}$ if and only if $P$ and $P'$ are operationally equivalent.  Then, from the definition of preparation non-contextuality, $\Pr(\lambda|P)\neq\Pr(\lambda|P')$ implies $\rho_{P}\neq\rho_{P'}$.  Thus the mapping $\mu(\lambda|\rho_{P})=\Pr(\lambda|P)$ is well-defined.

Now we show the mapping $\mu$ is affine.  That is,
\[
\mu(\lambda|p\rho_P+(1-p)\rho_{P'})=p\mu(\lambda|\rho_P)+(1-p)\mu(\lambda|\rho_{P'})
\]
for all $p\in[0,1]$ and all $P,P'\in D(\mathcal P)$.  There is some $P''\in D(\mathcal P)$ such that $\rho_{P''}=p\rho_P+(1-p)\rho_{P'}$. That is to say, $P''=P\vee P'$ while $\Pr(P)=p$ and $\Pr(P')=1-p$.  From the non-contextuality assumption this implies
\begin{align*}
\Pr(\lambda|P'')&=\Pr(\lambda|P\vee P')\\
&=\Pr(P)\Pr(\lambda|P)+\Pr(P')\Pr(\lambda|P').
\end{align*}
Applying the definition of $\mu$ to this yields
\[\mu(\lambda|\rho_{P''})=p\mu(\lambda|\rho_P)+(1-p)\mu(\lambda|\rho_{P'}),\]
which proves $\mu$ is affine.
\end{proof}

\begin{lemma}\label{convexity lemma meas}
Suppose the range of $\rho$ contains a basis (for $\herm$) and the ontological model $(\mathcal P, \mathcal M, K, \Lambda, \{\Pr(\lambda|P)\}, \{\Pr(k|M\wedge \lambda)\})$ is measurement non-contextual.  Then, there exists a unique affine mapping $\xi:\eff\to 
(K)$ satisfying $\Pr(k|M\wedge \lambda)=\xi( E_{M,k})$.
\end{lemma}

The proof of Lemma \ref{convexity lemma meas} is similar to that of Lemma \ref{convexity lemma prep}.  Together, the mappings $\mu$ and $\xi$ are called a \emph{classical representation of quantum theory}.
\begin{lemma}\label{lemma_classical}
A classical representation satisfies, for all $\lambda\in\Lambda$, all $\rho\in\den$ and all $E\in\eff$,
\begin{enumerate}[(a)]
\item $\mu_\rho(\lambda)\in[0,1]$ and $\sum_\lambda \mu_\rho(\lambda)=1$,
\item $\xi_E(\lambda)\in[0,1]$ and $\xi_\id(\lambda)=1$,
\item $\Tr(\rho E)=\sum_\lambda  \mu_\rho(\lambda) \xi_E(\lambda)$.
\end{enumerate}
\end{lemma}
In Reference \cite{Ferrie2008Frame} the name ``classical representation'' was defined to be a set of mappings satisfying (a)-(c).  Through Lemmas \ref{convexity lemma prep} and \ref{convexity lemma meas} we have shown that the assumption of non-contexuality guarantees that these mappings are well defined, convex linear and satisfy (a)-(c).  However, regardless of how we choose to arrive at a pair of convex linear mappings satisfying (a)-(c), one (or more) of the assumptions we make will be false as shown by the following theorem.
\begin{theorem}\label{theorem:noclassical}
A classical representation of quantum theory does not exist.
\end{theorem}
This theorem is implied by the equivalence of a classical representation and a non-contextual ontological model of quantum theory \cite{Spekkens2008Negativity} and the impossibility of the latter \cite{Spekkens2005Contextuality}.  This theorem was proven in reference \cite{Ferrie2008Frame} using the notion of a frame and its dual, which we will see later.  A direct and more intuitive proof was given in reference \cite{Ferrie2009Framed}.  The theorem is also implied by an earlier result on a related topic in reference \cite{Busch1993On}.

Now let us relax some assumptions in order to avoid a contradiction.  First, we will relax the assumption of non-contextuality.
\begin{definition}
A \emph{contextual representation of quantum theory} is a pair of mappings $\mu:\den\times C_{\mathcal P}\to P(\Lambda)$ and $\xi:\eff\times C_{\mathcal M}\to P(\Lambda)$ which satisfy, for all $\lambda\in\Lambda$, all $\rho\in\den$ and all $E\in\eff$, all $c_\mathcal P\in C_\mathcal P$, and all $c_\mathcal M\in C_\mathcal M$,
\begin{enumerate}[(a)]
\item $\mu_{\rho,c_\mathcal P}(\lambda)\in[0,1]$ and $\sum_\lambda \mu_{\rho,c_\mathcal P}(\lambda)=1$,
\item $\xi_{E,c_\mathcal M}(\lambda)\in[0,1]$ and $\xi_{\id,c_\mathcal M}(\lambda)=1$,
\item $\Tr(\rho E)=\sum_\lambda  \mu_{\rho,c_\mathcal P}(\lambda) \xi_{E,c_\mathcal M}(\lambda)$.
\end{enumerate}
Here $C_{\mathcal P}$ and $C_{\mathcal M}$ are the preparation and measurement \emph{contexts}.
\end{definition}

To the best of our knowledge such a construction does not exist in the literature.  In the previous section we looked at the Beltrametti-Bugajski model\footnote{These comments apply equally to other ontological models for pure states, such as the de Broglie-Bohm model.} which was non-contextual for pure states.  The model cannot be extended to include all states such that it remains non-contextual.  We could artificially create a \emph{preparation contextual} model by assigning a unique context to each convex decomposition of every density operator.  It is unclear whether such a construction could be useful.  It is more common, however, is to relax the assumption of non-negative probability rather than non-contextuality.
\begin{definition}\label{def:qp rep quant theory}
A \emph{quasi-probability representation of quantum theory} is a pair of affine mappings $\mu:\den\to L(\Lambda)$ and $\xi:\eff\to L(\Lambda)$ which satisfy, for all $\lambda\in\Lambda$, all $\rho\in\den$ and all $E\in\eff$,
\begin{enumerate}[(a)]
\item $\mu_\rho(\lambda)\in\mathbb R$ and $\sum_\lambda \mu_\rho(\lambda)=1$,
\item $\xi_E(\lambda)\in\mathbb R$ and $\xi_\id(\lambda)=1$,
\item $\Tr(\rho E)=\sum_\lambda  \mu_\rho(\lambda) \xi_E(\lambda)$.
\end{enumerate}
\end{definition}

It is immediately clear that theorem \ref{theorem:noclassical} is equivalent to the following ``negativity theorem'':
\begin{theorem}\label{theorem:neg}
A quasi-probability representation of quantum theory must have negativity in either its representation of states or measurements (or both).
\end{theorem}

There are many instances of mappings $\mu$ satisfying the first requirement.  In reference \cite{Ferrie2008Frame} it was shown by construction how to find a mapping $\xi$ which, together with $\mu$, satisfy all of them.  Most of the mappings $\mu$ are called \emph{phase space} functions as they conform to added mathematical structure not required in definition \ref{def:qp rep quant theory}.   A phase space representation is
then a particular type of quasi-probability representation which we formally define as follows:
\begin{definition}\label{def:phase_space_states}
If there exists a non-trivial symmetry group on $\Lambda$, $G$, carrying a
unitary representation $  U:G\to\unit$ and  a
quasi-probability representation satisfying the covariance property
$  U_g\rho  U_g^\dag\mapsto \{\mu_\rho(g(\lambda))\}_{\lambda\in\Lambda}$ for
all $\rho\in\den$ and $g\in G$, then $  \rho\mapsto \mu_\rho(\lambda)$ is a \emph{phase space
representation of quantum states}.
\end{definition}

It will turn out that all but a few of the representations reviewed below satisfy this stronger criterion.
\section{Quasi-probability in infinite dimensional Hilbert space\label{S:Infinite}}

Here we will review the quasi-probability distributions which have been defined for quantum states living in an infinite dimensional Hilbert space - the canonical example being a particle moving in one dimension.  Since there are a myriad of excellent reviews of the Wigner function and other phase space distributions \footnote{See reference \cite{Hillery1984Distribution} for a classic and reference \cite{Lee1995Theory} for a more recent review of phase space quasi-probability distributions.}, our discussion of them will be brief.  We will mainly focus on those details which have inspired analogous methods for finite dimensional Hilbert spaces.

First we start with the familiar Wigner function in section \ref{S(infinite):Wigner}.  The other phase space distributions, such as the Husimi function, are bundled up in section \ref{S(infinite):other}.

\subsection{Wigner phase space representation\label{S(infinite):Wigner}}

The position operator $Q$ and momentum operator $P$ are
the central objects in the abstract formalism of infinite
dimensional quantum theory.  The operators satisfy the canonical
commutation relations
\begin{equation*}
[  Q,   P]=i.\label{CCR}
\end{equation*}
We are looking for a joint probability distribution $\mu_\rho(p,q)$ for the state of the quantum system.  From the postulates of quantum mechanics we have a rule for calculating expectation values.  In particular, we can compute the characteristic function
\begin{equation*}\label{Wigner characteristic function}
\phi(\xi,\eta):=\av{e^{i(\xi q+\eta p)}}=\Tr(e^{i(\xi Q+\eta P)}\rho).
\end{equation*}
Since the characteristic function is just the Fourier transform of the joint probability distribution, we simply invert to obtain
\begin{equation}\label{Wigner Fourier invert of characteristic}
\mu_\rho(p,q)=\frac{1}{(2\pi)^2}\iint_{\mathbb R^2} \Tr(e^{i(\xi Q+\eta P)}\rho) e^{-i(\xi q +\eta p)}d\xi d\eta,
\end{equation}
which is the celebrated \emph{Wigner function} of $\rho$ \cite{Wigner1932On}.   The Wigner function is both positive and negative in general.  However, it otherwise behaves as a classical probability density on the classical phase space.  For these reasons, the Wigner function and others like it came to be called \emph{quasi-probability} functions.

The Wigner function is the unique representation satisfying the properties \cite{Bertrand1987A}
\begin{enumerate}[(a)]
\item[Wig(1)] For all $ \rho$, $\mu_{\rho}(q,p)$ is real.
\item[Wig(2)] For all $ \rho_1$ and $ \rho_2$, $$\Tr( \rho_1 \rho_2)=2\pi \int_{\mathbb R^2} dq dp\; \mu_{\rho_1}(q,p)\mu_{\rho_2}(q,p).$$
\item[Wig(3)] For all $ \rho$, integrating $\mu_{\rho}$ along the line $a  q+\ b p=c$ in phase space yields the probability that a measurement of the observable $a   Q+ b  P$ has the result $c$.
\end{enumerate}
We can write the Wigner function as
\begin{equation*}\label{Wigner as frame rep}
\mu_\rho(p,q)=\Tr\left[ F(p,q) \rho\right],
\end{equation*}
where
\begin{equation}
F(q,p):=\frac{1}{(2\pi)^2}\int_{\mathbb R^2} d\xi d\eta
\;e^{i\xi(Q-  q)+i\eta(P-  p)}.\label{Wigner_phasepoints}
\end{equation}
Thus the properties Wig(1)-(3) can be transformed into properties on a set of operators $  F(q,p)$ which uniquely specify the set in Equation \eqref{Wigner_phasepoints}.  These properties are
\begin{enumerate}[(a)]
\item[Wig(4)] $  F(q,p)$ is Hermitian.
\item[Wig(5)] $2\pi\Tr(F(q,p)F(q',p'))=\delta(q-q')\delta(p-p')$.
\item[Wig(6)] Let $P_c$ be the projector onto the eigenstate of $aQ+bP$ with eigenvalue $c$.  Then, $$\int_{\mathbb R^2} dq dp\; F(q,p)\delta(aq+bp-c)=P_c.$$
\end{enumerate}
These six properties are often the basis for generalizing the Wigner function to finite dimensional Hilbert spaces, as we will soon see.

\subsubsection{Applications: quantum teleportation}

The applications of the Wigner function are far reaching and not limited to to physics \footnote{For example, see reference \cite{Dragoman2005Applications} for a recent review of the applications of the Wigner function in signal processing.}.  A modern application can be found in reference \cite{Caves2004Classical} where Caves and W\'{o}dkiewicz use the Wigner function to obtain a hidden variable model of the continuous-variable teleportation protocol \cite{Vaidman1994Teleportation,Braunstein1998Teleportation}.  Later, in section \ref{S(finite):Cahill-Glauber}, we will discuss the much simpler discrete-variable teleportation protocol.  Here, then, we will avoid the details of the protocol and focus on the result.  It suffices to know the following: there are three quantum systems; the goal of the protocol is to transfer a quantum state from system 1 to system 3; the transfer is mediated through the special correlations between system 2 and system 3.  The success of the protocol is measured by the average \emph{fidelity}: a measure of the closeness of the initial state $\rho=\op\psi\psi$ of system 1 and the average final state $\rho_\textrm{out}$ of system 3.

Following Caves and W\'{o}dkiewicz we define $\nu=q+ip$ and index the Wigner function as $\mu_\rho(\nu)$.  This is convenient since the protocol is tailored to a quantum optical implementation where the outcomes of measurements are usually expressed as complex numbers.  Initially, the state of system 2 and 3 is described by the joint Wigner function $\mu_{2,3}(\alpha,\beta)$.

In terms of the Wigner functions, the average fidelity is
\[
\mathcal F=\pi\int d^2\nu d^2\beta \mu_\rho(\nu)\mu_{\rho_\textrm{out}}(\beta).
\]
This intuitively measures closeness by quantifying the overlap of the Wigner functions on the classical phase space.  The output state, determined by the details of the protocol, is
\[
\mu_{\rho_\textrm{out}}(\beta)=\int d^2\nu d^2\alpha\mu_{\rho}(\beta-\nu)\mu_{2,3}(\alpha,\nu-\conj\alpha).
\]
The initial Wigner function $\mu_\rho(\nu)$ and the joint Wigner function $\mu_{2,3}(\alpha,\beta)$ are determined by the particular implementation of the protocol.  The standard quantum optical implementation is done using coherent states of light.  It is easy to show that such states have positive Wigner functions \footnote{More difficult is to show that coherent states are the \emph{only} states with positive Wigner functions \cite{Hudson1974When}.}.  Thus, the Wigner function provides a classical phase space picture of the entire protocol.

A first step toward performing an experiment requiring genuine quantum resources might be to avoid the above classical description by teleporting a non-coherent quantum state.  Caves and W\'{o}dkiewicz have devised a classical explanation for this case as well.  The new model involves a randomization procedure which transforms the initial non-coherent state into a coherent one thus giving it a positive Wigner function.  However, it can be shown that within such a model, the fidelity is bounded: $\mathcal F<2/3$.  So $2/3$ emerges as a ``gold-standard'' since teleporting a non-coherent state with fidelity $\mathcal F\geq 2/3$ avoids this classical phase space description.

\subsection{Other phase space representations\label{S(infinite):other}}

Another class of solutions to the ordering problem is the association
$e^{i\xi q+i\eta p}\mapsto e^{i\xi   Q+i\eta   P}f(\xi,\eta)$ for some
arbitrary function $f$ \footnote{This is very closely related to the $s$-ordered Cahill-Glauber formalism \cite{Cahill1969Density}.  See Table 1 of \cite{Lee1995Theory} for a concise review of the traditional choices for $f$.}.

Consider again the classical particle phase space $\mathbb R^2$ and the continuous set of operators
\begin{equation}
  F(q,p):=\frac{1}{(2\pi)^2}\int_{\mathbb R^2} d\xi d\eta
\;e^{i\xi(q-  Q)+i\eta(p-  P)}f(\xi,\eta).\label{Wigner_phasepoints}
\end{equation}
The $f$ dependent distributions
\begin{equation}\label{f ordered distribution}
\mu^f_\rho(q,p):=\Tr( \rho  F(q,p))
\end{equation}
define quasi-probability functions on phase space alternative to the Wigner function, which is simply the $f=1$ special case of this more general formalism.

Besides the Wigner function, the most popular choices of $f$ are
\[
f(\xi,\eta)=e^{\pm\frac14(\xi^2+\eta^2)},
\]
which give, via equation \eqref{f ordered distribution}, the Glauber-Sudarshan \cite{Glauber1963Coherent,Sudarshan1963Equivalence} and Husimi \cite{Husimi1940Some} functions, respectively.  These two mappings are sometimes referred to as the P- and Q-representations (not to be confused with position and momentum representations).  We will follow the usual convention by introducing the annihilation operator
\begin{equation}
a=\frac{1}{\sqrt2}(Q+iP)\label{annihilation}
\end{equation}
and the \emph{coherent states} defined via $a\ket\alpha=\alpha\ket\alpha$, where we write $\alpha=q+ip$.  Then the Husimi function can be conveniently written
\begin{equation}\label{Husimi function}
Q(\alpha):=\mu^f_\rho(q,p)=\frac{1}{\pi}\ip{\alpha}{\rho \alpha}.
\end{equation}
The Glauber-Sudarshan function $\rho\to P(\alpha)$ can be expressed implicitly through the identity
\begin{equation}\label{Glauber-Sudarshan function}
\rho=\int d^2\alpha P(\alpha) \op{\alpha}{\alpha},
\end{equation}
where $d^2\alpha=(1/2)dqdp$.  Notice that this immediate implies the following \emph{duality} condition between the P- and Q-representation:
\begin{equation}\label{P-Q duality}
\Tr(\rho\rho')=\int d^2\alpha P_\rho(\alpha) Q_{\rho'}(\alpha).
\end{equation}
These functions and further discussed in section \ref{S:negativity_ss}.

\subsubsection{Application: quantumness witness}

Here we will discuss a more recent application of the P and Q functions of interest in quantum information and foundations \footnote{The P and Q functions are powerful visualization tools prominently used in the areas of quantum optics and quantum chaos \cite{Takahashi1989Distribution}.  See \cite{Schleich2001Quantum} for an overview of the applications in quantum optics.}.  We will concern ourself with a particular notion of ``non-classicality'' defined in reference \cite{Korbicz2005Hilberts,Alicki2008Quantumness,Alicki2008A}.

Consider two observables represented by the self-adjoint operators $R$ and $S$.  Another observable $W(R,S)$ written as an ordered power series of $R$ and $S$ is a \emph{quantumness witness} if it possesses at least one negative eigenvalue and the function $w(r,s)$ obtained by replacing $R$ and $S$ with its spectral elements is positive: $w(r,s)\geq0$ for all $r\in\spec(R)$ and $s\in\spec(S)$ ($\spec(\cdot)$ denotes the spectrum of the observable).

Now consider $R=Q$ and $S=P$, the usual position and momentum operators.  Recalling the parameterization in terms of the annihilation operator in equation \eqref{annihilation}, we define
\[
W(a):=\sum_{m,n} c_{mn} (a^\dag)^ma^n,
\]
for any choice of coefficients $c_{mn}$, such that
\[
w(\alpha)=\sum_{m,n} c_{mn} \conj \alpha^m\alpha^n\geq0
\]
and $W$ possesses at least one negative eigenvalue (an example of such a construction is $W(a)=(a^\dag)^2a^2-2ma^\dag a+m^2$ for $m\geq1$).  Then, in terms of the Q and P-representation defined above, we have
\begin{equation}\label{quantumness witness}
\av W =\Tr(\rho W)=\int d^2\alpha \ip{\alpha, W\alpha} P(\alpha)=\int d^2\alpha w(\alpha) P(\alpha).
\end{equation}
In optics especially, coherent states are considered classical.  From equation \eqref{Glauber-Sudarshan function} we see that if $P(\alpha)\geq0$, the quantum state is a statistical mixture of coherent states and hence just as classical.  So if $\rho$ is a classical state, equation \eqref{quantumness witness} tells us $\av W\geq0$.  Therefore, if we measure $\av W\leq0$, we can rule out the classical model of statistical mixtures of coherent states; we can say $W(a)$ detects the quantumness of the states.

\section{Quasi-probability in finite dimensional Hilbert space\label{S:Finite}}

Nearly all definitions of quasi-probability distributions for finite dimensional Hilbert spaces have been motivated by the Wigner function.  The earliest such effort was by Stratonovich and is reviewed in section \ref{S(finite):Sphere}.  The Stratonovich phase space is a sphere and hence continuous.  Later, many authors have define Wigner function analogs on \emph{discrete} phase spaces.  A sampling, with a bias towards those which have found application in quantum information theory, is given in sections \ref{S(finite):Wootters}-\ref{S(finite):Cahill-Glauber}.

There also exist quasi-probability distributions which where introduced to solve various problems far removed from proposing a finite dimensional analog of the Wigner function.  Sections \ref{S(finite):tables}-\ref{S(finite):SIC} review those quasi-probability distributions which do not have a canonical phase space structure and hence form a somewhat weaker analogy to the Wigner function.

We note that there exists many other quasi-probability distributions defined on discrete phase spaces which are not reviewed here \cite{Hannay1980Quantization,Cohen1986Joint,Feynman1987Negative,Galetti1988An,Opatrny1996Parametrized,Luis1998Discrete,Hakioglu1998Finitedimensional,Rivas1999The,Mukunda2004Wigner,Vourdas2004Quantum,Chaturvedi2005Wigner,Chaturvedi2006WignerndashWeyl,Gross2006Hudsons,Gross2007Nonnegative}.

\subsection{Spherical phase space\label{S(finite):Sphere}}

Here we will be concerned with a set of postulates put forth by Stratonovich \cite{Stratonovich1957On}.   The aim of Stratonovich was to find a Wigner function type mapping, analogous to that of a infinite dimensional system on $\mathbb R^2$, of a $d$ dimensional system on the sphere $\mathbb S^2$.  The first postulate is linearity and is always satisfied if the Wigner functions on the sphere satisfy
\begin{equation}
\mu_\rho(\mathbf{n})=\Tr( \rho  \triangle(\mathbf{n})),\label{DWF_sphere}
\end{equation}
where $\mathbf{n}$ is a point on $\mathbb S^2$.  The remaining postulates on this quasi-probability mapping are
\begin{align}
&\mu_\rho(\mathbf{n})^\ast=\mu_\rho(\mathbf{n}),\nonumber\\
&\frac{d}{4\pi}\int_{\mathbb S^2} d\mathbf{n} \;\mu_\rho(\mathbf{n})=1,\nonumber\\
&\frac{d}{4\pi}\int_{\mathbb S^2} d\mathbf{n} \;\mu_{\rho_1}(\mathbf{n})\mu_{\rho_2}(\mathbf{n})=\Tr( \rho_1 \rho_2),\nonumber\\
&\mu_{(g\cdot\rho)}(\mathbf{n})=\mu_\rho(\mathbf{n})^g,\;g\in\textrm{SU}(2),\nonumber
\end{align}
where $g\cdot\rho$ is the image of $  U_g \rho  U_g^\dag$ and $  U:\textrm{SU}(2)\to \mathbb{U}{\mathcal H}$ is an irreducible unitary representation of the group $\textrm{SU}(2)$.  These postulates are analogous to Wig(1)-(3) for the Wigner function modulo the second normalization condition (which could have be included in the Wigner function properties).

The continuous set of operators $ \triangle(\mathbf{n})$ is called a \emph{kernel} and we note it plays the role of the more familiar \emph{phase space point operators} in the latter.  Requiring that Equation \eqref{DWF_sphere} hold changes the postulates to new conditions on the kernel
\begin{align}
& \triangle(\mathbf{n})^\dag= \triangle(\mathbf{n}),\label{kernel_1}\\
&\frac{d}{4\pi}\int_{\mathbb S^2} d\mathbf{n} \; \triangle(\mathbf{n})= \id,\label{kernel_2}\\
&\frac{d}{4\pi}\int_{\mathbb S^2} d\mathbf{n} \;\Tr( \triangle(\mathbf{n}) \triangle(\mathbf{m})) \triangle(\mathbf{n})= \triangle(\mathbf{m}),\label{kernel_3}\\
& \triangle(g\cdot\mathbf{n})=  U_g \triangle(\mathbf{n})  U_g^\dag,\;g\in\textrm{SU}(2).\label{kernel_4}
\end{align}
These postulates are the spherical analogies of properties Wig(4)-(6) (again, modulo the normalization condition).  Defining $s=\frac{d-1}{2}$ as the \emph{spin}, 
Heiss and Weigert \cite{Heiss2000Discrete} provided a concise derivation of $2^{2s}$ unique kernels satisfying these postulates \footnote{This was also shown earlier \cite{Varilly1989The, Brif1999Phasespace} - see also references \cite{Arecchi1972Atomic, Scully1994Spin}.}.  They are
\begin{equation}\label{cont_kernel}
 \triangle(\mathbf{n})=\sum_{m=-s}^{s}\sum_{l=0}^{2s}\epsilon_l \frac{2l+1}{2s+1}C^{s\;l\;s}_{m\;0\;m} \phi_m(\mathbf{n})\phi_m^\ast(\mathbf{n}),
\end{equation}
where $C$ denotes the Clebsch-Gordon coefficients. Here, $\phi_m(\mathbf{n})$ are the eigenvectors of the operator $\mathbf{  S}\cdot\mathbf{n}$, where $\mathbf{  S}=(  X,  Y,  Z)$; and $\epsilon_l=\pm1$, for $l=1\ldots 2s$ and $\epsilon_0=1$.

Heiss and Weigert relax the postulates Equations \eqref{kernel_1}-\eqref{kernel_4} on the kernel $ \triangle(\mathbf{n})$ to allow for a pair of kernels $ \triangle^\mathbf{n}$ and $ \triangle_\mathbf{m}$.  The pair individually satisfy Equation \eqref{kernel_1}, while one of them satisfies Equation \eqref{kernel_2} and the other Equation \eqref{kernel_4}.  Together, the pair must satisfy the generalization of Equation \eqref{kernel_3}
\begin{equation}
\frac{d}{4\pi}\int_{\mathbb S^2} d\mathbf{n}\Tr( \triangle^\mathbf{n} \triangle_\mathbf{m}) \triangle^\mathbf{n}= \triangle_\mathbf{m}.\label{kernel_dual}
\end{equation}
A pair of kernels, together satisfying Equation \eqref{kernel_dual}, is given by
\begin{align*}
 \triangle_\mathbf{n}&=\sum_{m=-s}^{s}\sum_{l=0}^{2s}\gamma_l \frac{2l+1}{2s+1}C^{s\;l\;s}_{m\;0\;m} \phi_m(\mathbf{n})\phi_m^\ast(\mathbf{n}),\\
 \triangle^\mathbf{n}&=\sum_{m=-s}^{s}\sum_{l=0}^{2s}\gamma_l^{-1} \frac{2l+1}{2s+1}C^{s\;l\;s}_{m\;0\;m} \phi_m(\mathbf{n})\phi_m^\ast(\mathbf{n}),
\end{align*}
where each $\gamma_l$ is a finite non-zero real number and $\gamma_0=1$.  The original postulates are satisfied when $\gamma_l=\gamma_l^{-1}\equiv\epsilon_l$.

The major contribution of reference \cite{Heiss2000Discrete} is the derivation of a \emph{discrete} kernel $ \triangle_\nu:= \triangle_{\mathbf{n}_\nu}$, for $\nu=1\ldots d^2$ which satisfies the discretized postulates
\begin{align}
& \triangle_\nu^\dag= \triangle_\nu,\label{discrete_kernel_1}\\
&\frac{1}{d}\sum_{\nu=1}^{d^2}  \triangle^\nu= \id,\label{discrete_kernel_2}\\
&\frac{1}{d}\sum_{\nu=1}^{d^2} \Tr( \triangle_\nu \triangle^\mu) \triangle_\nu= \triangle^\mu,\label{discrete_kernel_3}\\
& \triangle_{g\cdot \nu}=  U_g \triangle_\nu  U_g^\dag,\;g\in\textrm{SU}(2).\label{discrete_kernel_4}
\end{align}
The subset of points $\mathbf{n}_\nu$ is called a \emph{constellation}.  The linearity postulate is not explicitly stated since it is always satisfied under the assumption
\begin{equation}
 \rho\to\mu_\rho(\nu)=\Tr( \rho  \triangle_\nu).\label{DWF_constellation}
\end{equation}
Equation \eqref{discrete_kernel_3} is called a \emph{duality} condition.  That is, it is only satisfied if $ \triangle_\nu$ and $ \triangle^\mu$ are \emph{dual bases} for $\herm$.  In particular,
\begin{equation*}
\frac{1}{d}\Tr( \triangle_\nu \triangle^\mu)=\delta_{\nu\mu}.
\end{equation*}
Although the explicit construction of a pair of discrete kernels satisfying Equations \eqref{discrete_kernel_1}-\eqref{discrete_kernel_4} might be computationally hard, their existence is a trivial exercise in linear algebra.  Indeed, so long as $ \triangle_\nu$ is a basis for $\herm$, its dual, $ \triangle^\mu$, is uniquely determined by
\begin{equation*}
\triangle^\mu=\sum_{\nu=1}^{d^2} \mathtt{G}^{-1}_{\nu\mu} \triangle_\nu,
\end{equation*}
where the Gram matrix $\mathtt G$ is given by
\begin{equation*}
\mathtt G_{\nu\mu}=\Tr( \triangle_\nu \triangle_\mu).
\end{equation*}
The authors of reference \cite{Heiss2000Discrete} note that almost any constellation leads to a discrete kernel $ \triangle_\nu$ forming a basis for $\herm$.  The term \emph{almost any} here means that a randomly selected discrete kernel will form, with probability 1, a basis for $\herm$.

\subsubsection{Application: NMR quantum computation}

The spherical quasi-probability functions for qubit systems ($d=2$) were put to use by Schack and Caves for the purpose of obtaining a classical model of nuclear magnetic resonance (NMR) experiments designed to perform quantum information tasks \cite{Schack1999Classical}.  For a single qubit we choose the kernels
\begin{align*}
\triangle_\mathbf{n}&=\frac{1}{2}(\id+\mathbf{n}\cdot\mathbf{\sigma}),\\
\triangle^\mathbf{n}&=\frac{1}{4\pi}(\id+3\mathbf{n}\cdot\mathbf{\sigma}),
\end{align*}
where $\mathbf{\sigma}=(X,Y,Z)$ are the usual Pauli operators.  In NMR experiments many qubits are employed to perform quantum information tasks such as error correction and teleportation.  Suppose there are $n$ qubits with total Hilbert space dimension $2^n$.  We choose an $n$-fold tensor product of the qubit kernels.  Explicitly, they are
\begin{align}
\triangle_\mathbf{n}&=\frac{1}{2^n}\bigotimes_{j=1}^n(\id+\mathbf{n}\cdot\mathbf{\sigma}),\label{qubit kernel 1}\\
\triangle^\mathbf{n}&=\frac{1}{(4\pi)^n}\bigotimes_{j=1}^n(\id+3\mathbf{n}\cdot\mathbf{\sigma}).\label{qubit kernel 2}
\end{align}
The quasi-probability function is given by
\[
\mu_\rho(\mathbf{n})=\Tr(\rho \triangle^\mathbf{n}).
\]
As expected, in general, this function is both positive and negative.

The quantum state of an NMR experiment is of the form
\begin{equation}\label{NMR state}
\rho=(1-\epsilon)\frac{1}{2^n}\id+\epsilon\rho_1,
\end{equation}
where $\rho_1$ is arbitrary but often chosen to be a specific pure state.  The parameter scales as
\[
\epsilon\propto \frac{n}{2^n}.
\]
So we have
\[
\mu_\rho(\mathbf{n})=\frac{1-\epsilon}{(4\pi)^n}+\epsilon\mu_{\rho_1}(\mathbf{n}).
\]
It is easy to determine the lower bound
\[
\mu_\rho(\mathbf{n})\geq \frac{1-\epsilon}{(4\pi)^n}-\frac{\epsilon 2^{2n-1}}{(4\pi)^n}.
\]
Thus, provided
\[
\epsilon\leq \frac{1}{1+2^{2n-1}},
\]
$\mu_\rho(\mathbf{n})\geq0$ and we have a representation of NMR quantum states in terms of classical probability distributions on a classical phase space.  Note however, that this definition of ``classical'' omits the reasonable requirement that it also be \emph{efficient}.  Indeed, the authors note that, for typical experimental values of the scaling parameter in $\epsilon$, such a classical representation is valid for $n<16$ qubits.

The spherical phase space representations have also been put to good use in visualizing decoherence \cite{Ghose2008Chaos} and photon squeezing \cite{Shalm2009Squeezing}.

\subsection{Wootters discrete phase space representation\label{S(finite):Wootters}}

In reference \cite{Wootters1987A}, Wootters defined a discrete
analog of the Wigner function.  Associated with each Hilbert space $\mathcal H$ of finite dimension $d$ is a
\emph{discrete phase space}.  First assume $d$ is
prime.  The \emph{prime phase space}, $\Phi$, is
a $d\times d$ array of points $\alpha=(q,p)\in \mathbb
Z_d\times\mathbb Z_d$.

A \emph{line}, $\lambda$, is the set of
$d$ points satisfying the linear equation $aq+bp=c$, where all
arithmetic is modulo $d$.  Two lines are \emph{parallel} if their
linear equations differ in the value of $c$.  The prime phase space $\Phi$ contains $d+1$ sets of $d$ parallel
lines called \emph{striations}.

Assume the the Hilbert space $\mathcal H$ has composite dimension
$d=d_1d_2\cdots d_k$.  The discrete phase space of the entire $d$
dimensional system is the Cartesian product of two-dimensional prime
phase spaces of the subsystems. The phase space is thus a $(d_1\times
d_1) \times (d_2\times d_2)\times\cdots(\times d_k\times d_k)$ array.
Such a construction is formalized as follows: the
\emph{discrete phase space} is the multi-dimensional array
$\Phi=\Phi_{1}\times\Phi_{2}\times\cdots\times\Phi_{k}$,
where each $\Phi_{i}$ is a prime phase space.  A \emph{point} is
the $k$-tuple $\alpha=(\alpha_1,\alpha_2, \ldots, \alpha_k)$ of
points $\alpha_i=(q_i,p_i)$ in the prime phase spaces.  A
\emph{line} is the $k$-tuple
$\lambda=(\lambda_1,\lambda_2,\ldots,\lambda_k)$ of lines in the
prime phase spaces.  That is, a line is the set of $d$ points
satisfying the equation
$$(a_1q_1+b_1p_1,a_2q_2+b_2p_2,\ldots,a_kq_k+b_kp_k)=(c_1,c_2,\ldots,c_k),$$
which is symbolically written $aq+bp=c$.  Two lines are
\emph{parallel} if their equations differ in the value
$c$.  As was the case for the prime phase spaces, parallel lines can be partitioned into sets, again called striations; the
discrete phase space $\Phi$ contains $(d_1+1)(d_2+1)\cdots(d_k+1)$ sets of $d$
parallel lines.

The construction of the discrete phase space is now been complete.
To introduce Hilbert space into the discrete phase space formalism,
Wootters chooses the following special basis for the space of Hermitian operators.  The set of operators $\{  A_\alpha:\alpha\in\Phi\}$ acting on an $d$ dimensional
Hilbert space are called \emph{phase point operators} if the operators satisfy
\begin{enumerate}[(a)]
\item[Woo(4)] For each point $\alpha$, $A_\alpha$ is Hermitian.
\item[Woo(5)] For any two points $\alpha$ and $\beta$, $\Tr(  A_\alpha  A_\beta)=d\delta_{\alpha\beta}$.
\item[Woo(6)] For each line $\lambda$ in a given striation, the operators
$  P_\lambda=\frac{1}{d}\displaystyle\sum_{\alpha\in\lambda}
A_\alpha$ form a projective valued measurement (PVM): a set of $d$
orthogonal projectors which sum to identity.
\end{enumerate}
Notice that these properties of the phase point operators Woo(4)-(6) are discrete analogs of the properties Wig(4)-(6) of the function $F$ defining the original Wigner function.  This definition suggests that the lines in the discrete phase
space should be labeled with states of the Hilbert space.  Since each striation is
associated with a PVM, each of the $d$ lines in a striation is
labeled with an orthogonal state.  For each $\Phi$, there is a unique set of phase point operators up
to unitary equivalence.

Although the sets of phase point operators are unitarily equivalent,
the induced labeling of the lines associated to the chosen set of
phase point operators are not equivalent.  This is clear from the
fact that unitarily equivalent PVMs do not project onto the same
basis.

The choice of phase point operators in reference \cite{Wootters1987A} will be adopted.
For $d$ prime, the phase point operators are
\begin{equation}
  A_\alpha=\frac{1}{d}\sum_{j,m=0}^{d-1}
\omega^{pj-qm+\frac{jm}{2}}  X^j  Z^m,\label{Wooters_phasepoint_prime}
\end{equation}
where $\omega$ is a $d$'th root of unity and $  X$ and $  Z$ are the generalized Pauli operators (see appendix \ref{appendix:Mathematical primer}). For
composite $d$, the phase point operator in $\Phi$ associated with
the point $\alpha=(\alpha_1,\alpha_2,\ldots,\alpha_k)$ is given by
\begin{equation}
  A_\alpha=  A_{\alpha_1}\otimes
A_{\alpha_2}\otimes\cdots\otimes
A_{\alpha_k},\label{Wootters_phasepoint_composite}
\end{equation}
where each $  A_{\alpha_i}$ is the phase point operator of
the point $\alpha_i$ in $\Phi_{i}$.

The $d^2$ phase point operators are linearly independent and form a
basis for the space of Hermitian operators acting on an $d$ dimensional
Hilbert space.  Thus, any density operator $  \rho$ can be
decomposed as
\[
  \rho=\sum_{q,p} \mu_\rho(q,p)  A(q,p),
\]
where the real coefficients are explicitly given by
\begin{equation}
\mu_\rho(q,p)=\frac{1}{d}\Tr(  \rho  A(q,p)).\label{rho2DWF}
\end{equation}
This discrete phase space function
is the Wootters \emph{discrete Wigner function}.  This discrete quasi-probability function satisfies the following properties which are the discrete analogies of the properties Wig(1)-(3) the original continuous Wigner function satisfies.
\begin{enumerate}[(a)]
\item[Woo(1)] For all $ \rho$, $\mu_{\rho}(q,p)$ is real.
\item[Woo(2)] For all $ \rho_1$ and $ \rho_2$, $$\Tr( \rho_1 \rho_2)= d \sum_{q,p} \mu_{\rho_1}(q,p)\mu_{\rho_2}(q,p).$$
\item[Woo(3)] For all $ \rho$, summing $\mu_{\rho}$ along the line $\lambda$ in phase space yields the probability that a measurement of the PVM associated with the striation which contains $\lambda$ has the outcome associated with $\lambda$.
\end{enumerate}

\subsubsection{Application: entanglement characterization}

In \cite{Franco2006Discrete}, Franco and Penna relate the negativity of Wootter's discrete Wigner function to entanglement.  Recall that a bipartite density matrix $\rho$ is \emph{separable} if it can be written as a convex combination of the form $$\rho=\sum_k p_k  \rho^{(1)}_k\otimes \rho^{(2)}_k,$$ for all $k$, where $ \rho^{(1)}_k$ and $\rho^{(2)}_k$ are states on the individual subsystems.  Let $\Phi_1$ and $\Phi_2$ be the DPS associated with $ \rho^{(1)}_k$ and $\rho^{(2)}_k$, respectively.

The Wootters representation of a density matrix of the form $\rho=\rho^{(1)}\otimes \rho^{(2)}$ is given by $\mu_\rho(\alpha)=\mu_\rho^{(1)}(\alpha_1)\mu_\rho^{(2)}(\alpha_2)$, where $\alpha=(\alpha_1,\alpha_2)\in\Phi_1\times\Phi_2$.
This can be shown as follows:
\begin{align*}
\mu_\rho(\alpha_1,\alpha_2)&=\frac{1}{d}\Tr( \rho^{(1)} A_{\alpha_1}\otimes \rho^{(2)} A_{\alpha_2})\\
&=\frac{1}{d}\sum_{\beta_1\in\Phi_1,\beta_2\in\Phi_2}\mu_\rho^{(1)}(\beta_1)\mu_\rho^{(2)}(\beta_2)\Tr( A_{\beta_1} A_{\alpha_1}\otimes A_{\beta_2} A_{\alpha_2})\\
&=\frac{1}{d}\sum_{\beta_1\in\Phi_1,\beta_2\in\Phi_2}\mu_\rho^{(1)}(\beta_1)\mu_\rho^{(2)}(\beta_2) d\delta_{\beta_1\alpha_1}\delta_{\beta_2\alpha_2}\\
&=\mu_\rho^{(1)}(\alpha_1)\mu_\rho^{(2)}(\alpha_2).
\end{align*}

Thus, separability can be recast entirely in terms of the discrete phase space.  That is, a discrete Wigner function is \emph{separable} if it can be written
\begin{equation}
\mu_\rho(\alpha)=\sum_k p_k\mu_\rho^{(1)}(\alpha_1)_k\mu_\rho^{(2)}(\alpha_2)_k,\label{separable_DWF}
\end{equation}
else it is \emph{entangled}.

The two qubit product state $\mu_\rho(\alpha)=\mu_\rho^{(1)}(\alpha_1)\mu_\rho^{(2)}(\alpha_2)$ with $\mu_\rho^{(1)}(\alpha_1)=\frac{1}{2}$ for some $\alpha_1$ and $\mu_\rho^{(2)}(\alpha_2)=\frac{1-\sqrt 3}{4}$ for some $\alpha_2$ will have the most negative value for a separable state, namely $\frac{1-\sqrt 3}{8}$.  Thus, if a two qubit Wigner function has a value strictly less that $\frac{1-\sqrt 3}{8}$, it is entangled.  Since entanglement is considered non-classical, negativity of the Wigner function (below some threshold) is associated with non-classicality.  However, even if a Wigner function is positive on all of phase space, it can still be entangled.  Therefore, Franco and Penna have found a new sufficient condition for entanglement in two qubits.

For a necessary condition, the authors of \cite{Franco2006Discrete} turn to the positive partial transpose condition \cite{Peres1996Separability, Horodecki1997Separability}.  The result is a two qubit state $\rho$ is separable if and only if both the discrete Wigner function of $\rho$ and the discrete Wigner function of $\rho^{\mathrm{T}_2}$ (the partial transpose) are non-negative everywhere on the discrete phase space.

Wootters discrete Wigner function has also found application in quantum teleportation \cite{Koniorczyk2001Wignerfunction}.  The authors have found the discrete phase space representation of the teleportation protocol much clearer especially when considering quantum systems with much larger than two dimensional Hilbert spaces.

\subsection{Extended discrete phase space\label{S(finite):Extended}}

In reference \cite{Cohendet1988A}, Cohendet \emph{et al} define a discrete analogue of the Wigner function which is valid for integer spin \footnote{This difficulty was overcome in a later paper \cite{Cohendet1990FokkerPlanck}.}.  That is, $\dim(\mathcal H)=d$ is assumed to be odd.  Whereas Wootters builds up a discrete phase space before defining a Wigner function, the authors of \cite{Cohendet1988A} implicitly define a discrete phase space through the definition of their Wigner function.

Consider the operators
\begin{equation*}
 W_{mn}\phi_k=\omega^{2n(k-m)}\phi_{k-2m},
\end{equation*}
with $m,n\in\mathbb Z_d$ and ${\phi_k}$ are the eigenvectors of $Z$.  Then, the \emph{discrete Wigner function} of a density operator $ \rho$ is
\begin{equation}\label{DWF_odd}
\mu_\rho^{\textrm{odd}}(q,p)=\frac{1}{d}\Tr( \rho  W_{qp}  P),
\end{equation}
where $  P$ is the parity operator.

The authors call the operators $ \triangle_{qp}=  W_{qp}  P$ \emph{Fano operators} and note that they satisfy
\begin{align*}
 \triangle_{qp}^\dag&= \triangle_{qp},\\
\Tr( \triangle_{qp} \triangle_{q'p'})&=d\delta_{qq'}\delta_{pp'},\\
  W_{xk}^\dag \triangle_{qp}  W_{xk}&= \triangle_{q-2x\;p-2k}.
\end{align*}
The Fano operators play a role similar to Wootters' phase point operators; they form a complete basis of the space of Hermitian operators.  The phase space implicitly defined through the definition of the discrete Wigner function \eqref{DWF_odd} is $\mathbb Z_d\times\mathbb Z_d$.  When $d$ is an odd prime, this phase space is equivalent to Wootters discrete phase space.  In this case the Fano operators are $ \triangle_{qp}=A_{(-q,p)}$.  This can seen by writing the Wootters phase point operators as
\begin{equation*}
A_{(q,p)}=\frac{1}{d}  X^{2q}  Z^{2p}  P \omega^{2qp}.
\end{equation*}

Let $\sigma\in\{\pm1\}$.  The \emph{extended} phase space is $\mathbb Z_d\times\mathbb Z_d\times\{\pm1\}$.  Define the new Wigner function
\[
\mu_\rho(q,p,\sigma)=\frac{1}{4d}\left(\frac{2}{d}+\sigma\mu_\rho^{\textrm{odd}}(q,p)\right).
\]
This function is satisfies the positivity and normalization requirements of a true probability distribution.

\subsubsection{Application: Master equation for an integer spin}

In the same paper, Cohendet \emph{et al} show the quantum dynamical equation of motion can be represented in the extended phase space as a classical stochastic process.  This is achieved by showing the time evolution of the discrete Wigner function is
\[
\frac{\partial}{\partial t}\mu_\rho(q,p,\sigma;t)=\sum_{q',p',\sigma'}A(q,p,\sigma|q',p',\sigma')\mu_\rho(q',p',\sigma';t),
\]
for a suitable choice of jump moments $A$.  This is in the form of the master equation of a Markov process.  The authors interpret this result as follows: ``Quantum mechanics of an integer spin appears as the mixture of two classical schemes of a spin.  However at random times the schemes are exchanged.''

\subsection{Even dimensional discrete Wigner functions\label{S(finite):Even}}

In reference \cite{Leonhardt1995QuantumState}, Leonhardt defines discrete analogues of the Wigner function for both odd and even dimensional Hilbert spaces.  In a later paper \cite{Leonhardt1996Discrete}, Leonhardt discusses the need for separate definitions for the odd and even dimension cases.   Naively applying his definition, or that of Cohendet \emph{et al}, of the discrete Wigner function for odd dimensions to even dimensions yields unsatisfactory results.  The reason for this is the discrete Wigner function carries redundant information for even dimensions which is insufficient to specify the state uniquely.  The solution is to enlarge the phase space until the information in the phase space function becomes sufficient to specify the state uniquely.

Suppose $\dim(\mathcal H)=d$ is odd.  Leonhardt defines the discrete Wigner function as
\begin{equation*}
\mu_\rho^{\textrm{Leo}}(q,p)=\frac{1}{d}\Tr( \rho  X^{2q}  Z^{2p}   P \omega^{2qp}).
\end{equation*}
Leonhardt's definition of an odd dimensional discrete Wigner function is unitarily equivalent to the Cohendet \emph{et al} definition.  That is, $\mu_\rho^{\textrm{Leo}}(q,p)=\mu_\rho^{\textrm{odd}}(-q,p)$.  To define a discrete Wigner function for even dimensions, Leonhardt takes half-integer values of $q$ and $p$.  This amounts to enlarging the phase space to $\mathbb Z_{2d}\times\mathbb Z_{2d}$.  Thus the \emph{even dimensional} discrete Wigner function is
\begin{equation*}
\mu_\rho^{\textrm{even}}(q,p)= \frac{1}{2d}\Tr( \rho  X^{q}  Z^{p}   P \omega^{\frac{qp}{2}}),
\end{equation*}
where the operators
\[
\triangle^{\textrm{even}}_{qp}=\frac{1}{2d}  X^{q}  Z^{p}   P \omega^{\frac{qp}{2}}
\]
could be called the even dimensional Fano or phase point operators.  Of course, these operators do not satisfy all the criteria which the Fano operators (in the case of Cohendet \emph{et al}) or the phase point operators (in the case of Wootters) satisfy; they are not orthogonal, for example.  Moreover, they are not even linearly independent which can easily be inferred since there are $4d^2$ of them and a set of linearly independent operators contains a maximum of $d^2$ operators.

\subsubsection{Application: quantum computation}

Leonhardt's discrete Wigner function has been used to visualize and gain insights for algorithms expected to be performed on a quantum computer \cite{Bianucci2002Discrete, Miquel2002Quantum, Miquel2002Interpretation}.  For each step in a quantum algorithm the state $\rho(t)$ of the quantum computer is update via some unitary transformation
\[
\rho(t+1)=U\rho(t) U^\dag.
\]
This can be represented in the discrete phase space as
\[
\mu_\rho(q,p;t)=\sum_{q'p'} Z(p,q|p',q') \mu_\rho(q',p';t),
\]
where $Z$ can be easily obtained from $U$.  This resembles the update map for the probabilities of classical stochastic variables.  However, the properties of $Z$ imply that not all admissible maps are classical; they do not connect single points in phase space and hence are ``nonlocal''.  In reference \cite{Miquel2002Quantum} the authors identify a family of classical maps which can be efficiently implemented on a quantum computer.  The authors admit that the ultimate usefulness of this approach is uncertain but speculate that the phase space representation may inspire improvement and innovation in quantum algorithms.  It certainly makes for some inspiring pictures!

The Leonhardt phase space formalism has also been applied to study decoherence in quantum walks \cite{Lopez2003Phasespace}.  For large system, numerics are often employed to study the main features.  The phase space method offers an intuitive and visual alternative.  It allows one to visually see the quantum interference and its disappearance under decoherence.  Related to these is a hybrid approach between the Wootters and Leonhardt discrete phase spaces used to analyze various aspects of quantum teleportation \cite{Paz2002Discrete}.

\subsection{Finite fields discrete phase space representation\label{S(finite):Fields}}
Recall that when $\dim(\mathcal H)=d$ is prime, Wootters defines the discrete phase space as a $d\times d$ lattice indexed by the group $\mathbb Z_d$.  In reference \cite{Wootters2004Picturing}, Wootters generalizes his original construction of a discrete phase space to allow the $d\times d$ lattice to be indexed by a finite field $\mathbb F_d$ which exists if and only if $d=p^n$ is an integer power of a prime number.  This approach is discussed at length in the paper \cite{Gibbons2004Discrete} authored by Gibbons, Hoffman and Wootters (GHW).

Similar to his earlier approach, Wootters defines the \emph{phase space}, $\Phi_d$, as
a $d\times d$ array of points $\alpha=(q,p)\in \mathbb
F_d\times\mathbb F_d$.  A \emph{line}, $\lambda$, is the set of
$d$ points satisfying the linear equation $aq+bp=c$, where all
arithmetic is done in $\mathbb F_d$.  Two lines are \emph{parallel} if their
linear equations differ in the value of $c$.

The mathematical structure of $\mathbb{F}_d$ is appealing because
lines defined as above have the following useful properties: (i)
given any two points, exactly one line contains both points, (ii)
given a point $\alpha$ and a line $\lambda$ not containing $\alpha$,
there is exactly one line parallel to $\lambda$ that contains
$\alpha$, and (iii) two nonparallel lines intersect at exactly one
point. Note that these are usual properties of lines in Euclidean
space.  As before, the $d^2$ points of the phase space $\Phi_d$ can be partitioned into
$d+1$ sets of $d$ parallel lines called \emph{striations}.  The
line containing the point $(q,p)$ and the origin $(0,0)$ is called a
\emph{ray} and consists of the points $(sq,sp)$, where $s$ is a
parameter taking values in $\mathbb{F}_d$.  We choose each ray,
specified by the equation $aq+bp=0$, to be the representative of the
striation it belongs to.

A translation in phase space, $\mathcal T_{\alpha_0}$, adds a constant
vector, $\alpha_0=(q_0,p_0)$, to every phase space point:
$\mathcal T_{\alpha_0}\alpha=\alpha+\alpha_0$.  Each line, $\lambda$, in a
striation is invariant under a translation by any point contained in
its ray, parameterized by the points $(sq,sp)$. That is,
\begin{equation}
\tau_{(sq,sp)}\lambda=\lambda.\label{striaeinv}
\end{equation}

The discrete Wigner function is
\begin{equation*}
\mu_\rho^{\textrm{field}}(q,p)=\frac{1}{d}\Tr( \rho A_{(q,p)}),
\end{equation*}
where now the Hermitian \emph{phase point operators} satisfy the following properties for a projector valued function $  Q$, called a \emph{quantum net}, to be defined later.
\begin{enumerate}[(a)]
\item[GHW(4)] For each point $\alpha$, $  A$ is Hermitian.
\item[GHW(5)] For any two points $\alpha$ and $\beta$, $\Tr(  A_\alpha  A_\beta)=d\delta_{\alpha\beta}$.
\item[GHW(6)] For any line $\lambda$, $\displaystyle\sum_{\alpha\in\lambda}A_\alpha=d   Q(\lambda)$.
\end{enumerate}
The projector valued function $  Q$ assigns quantum states to lines in phase space.  This mapping is required to satisfy the special property of \emph{translational covariance}, which is defined after a short, but necessary, mathematical digression.  Notice first that properties GHW(4) and GHW(5) are identical to Woo(4) and Woo(5).  Also note that if GHW(6) is to be analogous to Woo(6), the property of translation covariance must be such that the set $\{  Q(\lambda)\}$ when $\lambda$ ranges over a striation forms a PVM.

The set of elements $E=\{e_0,...,e_{n-1}\}\subset \mathbb F_d$ is called a
\emph{field basis} for $\mathbb{F}_d$ if any element, $x$, in
$\mathbb{F}_d$ can be written
\begin{equation}
x=\sum_{i=0}^{n-1}x_ie_i,\label{fieldbasis}
\end{equation}
where each $x_i$ is an element of the prime field $\mathbb{Z}_p$.
The \emph{field trace}\footnote{Note that we will distinguish the
field trace, $\textrm{tr}(\cdot)$, from the usual trace of a Hilbert
space operator, $\textrm{Tr}(\cdot)$, by the case of the first
letter.} of any field element is given by
\begin{equation}
\textrm{tr}(x)=\sum_{i=0}^{n-1}x^{p^i}.\label{fieldtrace}
\end{equation}
There exists a unique field basis, $\tilde E =\{\tilde e_0,...\tilde
e_{n-1}\}$, such that $\textrm{tr}(\tilde e_i e_j)=\delta_{ij}$.  We
call $\tilde E$ the \emph{dual} of $E$.

The construction presented in reference \cite{Gibbons2004Discrete} is physically significant for a
system of $n$ objects (called \emph{particles}) having a $p$
dimensional Hilbert space. A translation operator, $  T_\alpha$
associated with a point in phase space $\alpha=(q,p)$ must act
independently on each particle in order to preserve the tensor
product structure of the composite system's Hilbert space.  We
expand each component of the point $\alpha$ into its field basis
decomposition as in Equation \eqref{fieldbasis}
\begin{equation}
q=\sum_{i=0}^{n-1}q_ie_i\label{q}
\end{equation}
and
\begin{equation}
p=\sum_{i=0}^{n-1}p_if\tilde e_i,\label{p}
\end{equation}
with $f$ any element of $\mathbb{F}_d$.  Note that the basis we
choose for $p$ is a multiple of the dual of that chosen for $q$.
Now, the translation operator associated with the point $(q,p)$ is
\begin{equation}
  T_{(q,p)}=\bigotimes_{i=0}^{n-1}  X^{q_i}  Z^{p_i},\label{transopt}
\end{equation}
Since $X$ and $Z$ are unitary, $  T_\alpha$ is unitary.

We assign with each line in phase space a pure quantum state.  The
quantum net $  Q$ is defined such that for
each line, $\lambda$, $  Q(\lambda)$ is the operator which projects
onto the pure state associated with $\lambda$.  As a consequence of the
choice of basis for $p$ in Equation \eqref{p}, the state assigned to the line
$\tau_\alpha\lambda$ is obtained through
\begin{equation}
  Q(\tau_\alpha\lambda)=  T_\alpha
  Q(\lambda)  T_\alpha^\dag.\label{transcov}
\end{equation}
This is the condition of translational covariance and it implies that each striation is associated with an orthonormal
basis of the Hilbert space.  To see this, recall the property in Equation
\eqref{striaeinv}. From Equation \eqref{transcov}, this implies that, for
each $s\in\mathbb{F}_d$, $  T_{(sq,sp)}$ must commute with
$  Q(\lambda)$, where the line $\lambda$ is any line in the
striation defined by the ray consisting of the points $(sq,sp)$.
That is, the states associated to the lines of the striation must be
common eigenstates of the unitary translation operator
$  T_{(sq,sp)}$, for each $s\in \mathbb{F}_d$. Thus, the states are
orthogonal and form a basis for the Hilbert space.  That is, their projectors form a PVM which makes GHW(6) identical to Woo(6) when $d$ is prime.

In reference \cite{Gibbons2004Discrete}, the authors note that, although the association
between states and vertical and horizontal lines is fixed, the
quantum net is not unique.  In fact, there are $d^{d+1}$ quantum
nets which satisfy Equation \eqref{striaeinv}.  When $d$ is prime, one of these quantum nets corresponds exactly to the original discrete Wigner function defined by Wootters in Section \ref{S(finite):Wootters}.

\subsubsection{Application: quantum computation}
As conjectured by Galv\~{a}o \cite{GalvAo2005Discrete}, the authors of reference \cite{Cormick2006Classicality} have shown the only quantum states having a non-negative discrete Wigner function \footnote{Note that it is assumed the discrete Wigner function is non-negative for all definitions - that is, for all quantum nets.} are convex combinations of stabilizer states, which are simultaneous eigenstates of the generalized Pauli operators \cite{Gottesman1997Stabilizer}.  Working only with stabilizer states is ``classical'' in the sense that that one can represent them with only a polynomial number of classical bits whereas an arbitrary quantum state requires a exponential number of bits \cite{Nielsen2000Quantum}.

Strengthening the connection between negativity and non-classicality, it was also shown that the unitary operators preserving the non-negativity of the discrete Wigner function are a subset of the Clifford group, which are those unitaries which preserve Pauli operators under a conjugate mapping.  According to the Gottesman-Knill theorem, a quantum computation using only operators from the Clifford group and stabilizer states can be efficiently simulated on a classical computer \cite{Gottesman1997Stabilizer}.  Thus, as noted in reference \cite{GalvAo2005Discrete} for a particular computational model, negativity of the Wigner function is necessary for quantum computational speedup.   A direct generalization of this result to continuous variables appears in reference \cite{Bartlett2002Efficient} where the authors generalize the stabilizer states to the familiar Gaussian states used in quantum optics settings.

This discrete Wigner function was also used to analyze quantum error correcting codes in reference \cite{Paz2005Qubits}.  The aim was to gain insights and intuition for various quantum maps by studying their pictorial representation in the discrete phase space.

\subsection{Discrete Cahill-Glauber formalism\label{S(finite):Cahill-Glauber}}

In reference \cite{Ruzzi2005Extended}, Ruzzi \emph{et al} have discretized the Cahill-Glauber phase-space formalism.  The set of operators $\{S(\eta,\xi)\}$, where $\eta,\xi\in[-l,l]$ and $l=\frac{d-1}{2}$ ($d$ odd), is called the \emph{Schwinger basis} and explicitly given by
\[
S(\eta,\xi)=\frac{1}{\sqrt{d}}X^\eta Z^\xi \omega^{\frac{\eta\xi}{2}}.
\]
These $d^2$ operators form an orthonormal basis for the space of linear operators.  In analogy with the Cahill-Glauber formalism, the basis is generalized to
\[
S^{(s)}(\eta,\xi)=S(\eta,\xi)\mathcal K(\eta,\xi)^{(-s)},
\]
where $\abs s\leq 1$ is any complex number and $\mathcal K(\eta,\xi)$ is a (relatively) complicated expression of Jacobi $\vartheta$-functions (see the Appendix of reference \cite{Ruzzi2005Extended}).  Next we take the Fourier transform
\[
T^{(s)}(q,p)=\frac{1}{\sqrt{d}}\sum_{\eta,\xi=-l}^l S^{(s)}(\eta,\xi)\omega^{-(\eta q + \xi p)}.
\]
The set operators $\{T^{(s)}(q,p)\}$ is the discrete analog of the $s$-ordered mapping kernel of the Cahill-Glauber formalism.  Moreover, the authors of reference \cite{Ruzzi2005Extended} have shown that the continuous limit of this set is indeed the Cahill-Glauber mapping kernel.

Suppose $s$ is real.  Then , the operators $\{T^{(s)}(q,p)\}$ enjoy the following familiar properties:
\begin{align*}
& T^{(s)}(q,p)^\dag= T^{(s)}(q,p),\\
& \Tr( T^{(s)}(q,p) T^{(-s)}(q',p'))=d\delta_{qq'}\delta_{pp'}.
\end{align*}
Thus, similarly to the discrete kernel of Heiss and Weigert, $\{T^{(s)}(q,p)\}$ and $\{T^{(-s)}(q,p)\}$ are dual bases for the space of Hermitian operators.

In the now familiar way, we can define a quasi-probability function on the $(q,p)$ phase space as
\begin{equation}\label{Ruzzi s-ordered function}
\mu^{(s)}_\rho(q,p)=\Tr(T^{(s)}(q,p) \rho).
\end{equation}
Cahill and Glauber showed, for their $s$-ordered formalism, that $s=0$ corresponds to the Wigner function; $s=1$ corresponds to the Husimi function; and $s=-1$ corresponds to the Glauber-Sundarshan function.  Using equation \eqref{Ruzzi s-ordered function}, we call, for example, the function obtained when $s=0$ the discrete Wigner function.

\subsubsection{Application: quantum teleportation}

Marchiolli \emph{et al} have applied this formalism to quantum tomography and teleportation \cite{Marchiolli2005Extended}.  The teleportation protocol was analyzed for arbitrary $s$ but, for brevity, we will consider the $s=0$ case (which is now assumed so the superscript can be ignored).  The teleportation protocol utilizes entanglement to transfer a quantum state between two parties through the exchange of only a small amount of classical information \cite{Nielsen2000Quantum}.  Consider the tripartite system
\[
\rho=\rho^{(1)}\otimes\rho^{(2,3)},
\]
where one party possess subsystem 1 and 2 and the other possess subsystem 3.  The goal is for $\rho^{(1)}$ to be transferred from subsystem 1 to subsystem 3 without simply swapping them.  It is essential that the shared state $\rho^{(2,3)}$ be entangled.  In particular, assume it is a maximally entangled pure Bell-state \cite{Nielsen2000Quantum}.  We choose, following Wootters \cite{Wootters1987A}, to construct the global phase space to be a Cartesian product of the phase spaces of the individual subsystems.  The discrete Wigner function of the whole system is then
\[
\mu_\rho(q_1,q_2,q_3;p_1,p_2,p_3)=\mu_{\rho^{(1)}}(q_1,p_1)\mu_{\rho^{(2,3)}}(q_2,q_3;p_2,p_3).
\]
A Bell-measurement is performed on the first two subsystems, which in phase space is interpreted as a measurement of the total momentum  and relative coordinate of the subsystem composed of subsystems 1 and 2.  Marginalizing over subsystems 1 and 2 gives
\[
\mu_{\rho^{(3)}}(q_3,p_3)=\mu_{\rho^{(1)}}(q_3-\alpha,p_3+\beta),
\]
where $\alpha$ and $\beta$ parameterize the result of the Bell-measurement (note $\rho^{(3)}$ can be identified as the reduced state of subsystem 3).  Thus, the final state of the subsystem 3 is simply a displacement in the phase space and communicating only the measurement result $(\alpha,\beta)$ leads to recovery of the initial state.

The discrete Husimi function ($s=1$) was used to define a discrete analog of squeezed states \cite{Marchiolli2007Discrete} and to analyze spin tunneling effects in a particular toy model of interacting fermions \cite{Marchiolli2009Quasiprobability}.

\subsection{Probability tables\label{S(finite):tables}}

In 1986, before introducing the discrete Wigner function, Wootters represented the quantum state as a ``probability table'' which was simply a list of outcome probabilities for a complete set of measurements  \cite{Wootters1986Quantum}.  The complete set of interest was that of \emph{mutually unbiased bases} (MUBs).  We call $n$ bases $\{\psi_k^n\}$ mutually unbiased if they satisfy
\begin{equation}\label{MUB definition}
\abs{\ip{\psi_{k'}^{n'}}{\psi_k^n}}^2=\delta_{kk'}\delta_{nn'}+\frac1d(1-\delta_{nn'}).
\end{equation}
Wootters noted for $d$ prime, a set of $n=d+1$ MUBs could be explicitly constructed via a prescription in reference \cite{Ivonovic1981Geometrical}.  Wootters also posed many questions of MUBs, some of which have now been answered.  It is now known that for any dimension $3\leq n\leq d+1$, where the upper bound can be achieved, by construction, for any dimension which is a power of a prime \footnote{For a recent review of the MUB problem see \cite{Combescure2006The}.}.

Here we will consider the case when $d$ is prime and all probability tables for non prime dimensions can be built up from those for their prime factors, in much the same was as was done in section \ref{S(finite):Wootters} for the discrete phase spaces.

Consider the generalized Pauli operator $Z$ and its eigenbasis $\{\phi_k\}$ and the projectors onto these vectors $P_k:=\phi_k\phi_k^*$.  Define the finite Fourier transform
\begin{equation}\label{finite fourier transform}
F=\frac{1}{\sqrt d}\sum_{k,k'=0}^{d-1} \omega^{kk'}\phi_k\phi_{k'}^*
\end{equation}
and the operator
\begin{equation}\label{MUB V operator}
V=\sum_{k=0}^{d-1}\omega^{\frac{k^2}{2}}FP_kF^\dag.
\end{equation}
Here, as before, division by two represents the multiplicative inverse of the element 2.  For Hilbert space dimension $d=2$, this operator requires the special definition
\begin{equation}\label{MUB V operator d=2}
V=\frac12\left(
    \begin{array}{cc}
      1+i & 1-i \\
      1+i & 1-i \\
    \end{array}
  \right).
\end{equation}
Now we can construct $d+1$ MUBs via
\begin{align*}
\psi^0_k&=\phi_k,\\
\psi_k^n&=V^n\phi_k,~n=1,\ldots,d.
\end{align*}
We will denote the projectors onto these basis vectors $P(n,k):=\psi_k^n\psi_k^{n*}$.  Then the probability of obtaining the $k$th outcome when measuring in the $n$th basis is
\begin{equation}\label{MUB probabilities}
\mu_\rho(n,k)=\Tr(\rho P(n,k)).
\end{equation}
This can be view as a matrix in which the columns index the measurements while the rows index the outcomes.  This can also be viewed as a mapping whose inverse is given by
\begin{equation}\label{MUB probability inversion}
\rho=\sum_{n=0}^d\sum_{k=0}^{d-1}\mu_\rho(n,k)P(n,k)-\id.
\end{equation}

\subsubsection{Application: quantum mechanics without amplitudes}

The purpose of reference \cite{Wootters1986Quantum} was not to introduce a new representation of the quantum state \emph{per se}, but to show that the whole of operational formalism of quantum mechanics can be done rather simply without complex numbers.

Wootters notes first:
\begin{quote}
It is obviously possible to devise a formulation of quantum mechanics
without probability amplitudes. One is never forced to use any quantities in
one's theory other than the raw results of measurements. However, there is
no reason to expect such a formulation to be anything other than
extremely ugly.
\end{quote}

To our surprise, the rule for transitioning between the probability tables turn out to be remarkably simple.  In quantum theory the transition probability from state $\rho$ to $\rho'$ is the probability of preparing the state $\rho$, performing the measurement $\{\rho',\id-\rho'\}$ and obtaining $\rho'$.  This transition probability is $\Pr(\rho\to\rho')=\Tr(\rho\rho')$.  If we work with the probability tables and call the tables $\mu$ and $\mu'$, Wootters obtains
\begin{equation}\label{prob tables update rule}
\Pr(\mu\to \mu')=\sum_{n=0}^d\sum_{k=0}^{d-1}\mu(n,k)\mu'(n,k)-1.
\end{equation}

Unfortunately, as Wootters notes, it is not easy to ignore the density matrix altogether.  We have yet to specify which probability tables are valid and which do not correspond to quantum states.  The simplest characterization of valid probability tables is to say those for which equation \eqref{MUB probability inversion} is a unit trace positive semi-definite matrix.  This is unsatisfying as we would like a characterization independent of the density matrix.

\subsection{Hardy's vector representation \label{S(finite):Hardy}}

In reference \cite{Hardy2001Quantum} Hardy showed that five axioms are sufficient to imply a
special vector representation which is equivalent to an
operational form of quantum theory.  We first describe the vector representation.

Consider a basis for a $d$ dimensional Hilbert space $\{\phi_k\}$ (the eigenbasis of $Z$, say) and the following set of $d^2$ projectors:
\begin{equation}\label{Hardy's basis}
P_{kj}:=\begin{cases}
\phi_k\phi_k^*& \textrm{if}~k=j\\
(\phi_k+\phi_j)(\phi_k+\phi_j)^* & \textrm{if}~k<j\\
(\phi_k+i\phi_j)(\phi_k+i\phi_j)^* & \textrm{if}~k>j
\end{cases}.
\end{equation}
These projectors span the space of linear operators on the Hilbert space spanned by $\{\phi_k\}$.  Now we vectorize by choosing an arbitrary but fixed ordering convention.  For definiteness, we choose to stack the rows on top of one another.  To this end, define $\alpha:=dk+j$ and $P(\alpha):=P_{kj}$.  Then, the vector representation of the state $\rho$ is given by
\begin{equation}\label{Hardy vector rep}
\mu_\rho(\alpha)=\Tr(\rho P(\alpha)).
\end{equation}
Now the outcome of any quantum measurement can be assigned a positive operator $E$.  Call this ``outcome $E$''.  Define the vector $\xi_E(\alpha)$ implicitly through
\[
E=\sum_\alpha \xi_E(\alpha) P(\alpha).
\]
Then, the probability of ``outcome $E$'' is given by
\[
\Pr(\textrm{``outcome $E$''})=\sum_\alpha \xi_E(\alpha) \mu_\rho(\alpha),
\]
which, in vector notation, we can write as the dot product $\overrightarrow{\xi}\cdot\overrightarrow{\mu}$.

We define the sets $M$ and $\Xi$ as the set of vectors obtainable through the mappings $\rho\mapsto\mu$ and $E\mapsto \xi$ defined above.  More precise statements, in the form of inequalities, which make no recourse to the usual quantum mechanical objects, can be made to define these sets.  Assuming this has been done, we can rephrase the axioms of quantum mechanics, without mention of Hermitian operators and the like, in this vector representation succinctly as follows: states are represent by vectors $\overrightarrow{\mu}\in M$; measurement outcomes are represented by vectors $\overrightarrow{\xi}\in\Xi$; the probability of ``outcome $\overrightarrow{\xi}$'' in state $\overrightarrow{\mu}$ is given by $\Pr(\textrm{``outcome $\overrightarrow{\xi}$''})=\overrightarrow{\xi}\cdot\overrightarrow{\mu}$.

\subsubsection{Application: quantum axiomatics}
As was the case in the previous section, this vector representation was not introduced as such.  In references \cite{Hardy2001Quantum, Hardy2001Why}, Hardy has shown that five axioms are sufficient to imply the real vector formalism of quantum mechanics.  The frequency interpretation of probability was given its own axiom.  However, if we take our everyday intuitive notion of probability \cite{PortaMana2007Studies}, we no longer require this first axiom, which is independent of the rest \cite{Schack2003Quantum}.

We will make use of the following definitions:
\begin{itemize}
\item The number of \emph{degrees of freedom}, $K$, is defined
as the minimum number of yes-no measurements whose outcome probabilities are
needed to determine the state (of belief in the mind of a reasonable agent), or, more
roughly, as the number of real parameters required
to specify the state.
\item The \emph{dimension}, $d$, is defined as the maximum
number of states that can be reliably distinguished
from one another in a single shot measurement.
\end{itemize}

Axiom 1 defines probability as limiting frequencies and is not required \cite{Schack2003Quantum}.  The remainder of the axioms are as follows:
\begin{itemize}
\item[2] \emph{Simplicity.} $K$ is determined by a function
of $N$ (i.e. $K = K(d)$) where $d = 1, 2, \ldots$ and
where, for each given $d$, $K$ takes the minimum
value consistent with the axioms.
\item[3] \emph{Subspaces.} A system whose state is constrained
to belong to an $n$ dimensional subspace
(i.e. have support on only $n$ of a set of $d$ possible
distinguishable states) behaves like a system
of dimension $n$.
\item[4] \emph{Composite systems.} A composite system
consisting of subsystems $A$ and $B$ satisfies $d =
d_Ad_B$ and $K = K_AK_B$.
\item[5] \emph{Continuity.} There exists a continuous reversible
transformation on a system between any
two pure states of that system.
\end{itemize}

These four axioms are sufficient for a derivation of the vector representation of quantum theory defined above.  This axiomatization is also important for contrasting quantum theory with classical probability theory.  As Hardy has shown, discrete classical probability theory (of dice, coins and so on) can be derived from only axioms 2, 3 and 4.  That is, the only difference between quantum and classical theory is the existence of a \emph{continuous} reversible transformation between pure states.

\subsection{The real density matrix\label{S(finite):Havel}}

In reference \cite{Havel2002The} Havel defined the ``real
density matrix'' which, not surprisingly, is a particular real-valued matrix representation of the quantum state.

For $d=2$, define the $2\times2$ matrix of Pauli operators as
\begin{equation}\label{pauli operator matrix 2}
P=\left(
    \begin{array}{cc}
      \id & X \\
      Y & Z \\
    \end{array}
  \right).
\end{equation}

For $d=2^n$, denote the bits in the binary expansion of $k$ as
\[
k=\sum_{a=1}^n k_a\times 2^{n-a},
\]
and similarly for $j$.  Then, the $d\times d$ matrix of Pauli operators is given by
\begin{equation}\label{pauli operator matrix d}
P_{kj}=P_{k_1j_1}\otimes\cdots\otimes P_{k_nj_n}.
\end{equation}
These $d^2$ operators are orthogonal, and hence form a basis for the space of linear operators on the $d$ dimensional Hilbert space.  Therefore, each density matrix can be expressed as
\[
\rho=\frac{1}{d}\sum_{k,j=0}^{d-1}\sigma_{kj}P_{kj},
\]
where the coefficients $\sigma_{kj}$, explicitly given by
\begin{equation}\label{real density matrix}
\sigma_{kj}=\Tr(\rho P_{kj}),
\end{equation}
form the \emph{real density matrix}.

\subsubsection{Application: NMR pedagogy}

Since the observables measured in NMR experiments are elements of the matrix of Pauli operators \eqref{pauli operator matrix d}, the elements of the real density matrix are the experimentally measurable values.  There is no need to reconstruct the density matrix.  This is also a convenient fix to the problem of reporting or visualizing a quantum state.  Since the density matrix contains $d^2$ complex values, it is often graphically displayed as two $d\times d$ matrices of the real and imaginary parts.  Not only is this redundant, it is conceptually awkward.  On the other hand, the real density matrix can be displayed as a single $d\times d$ matrix of real values.  Havel offers the real density matrix as useful teaching device in such situations.

\subsection{Symmetric representations\label{S(finite):SIC}}
Consider the unitary group
\begin{equation}
U_{(p,q)}=\omega^{\frac{pq}{2}}X^pZ^q,\label{Weylopts}
\end{equation}
where $(p,q)\in\mathbb Z_d\times\mathbb Z_d$. In reference \cite{Renes2004Symmetric}, the authors conjecture\footnote{Apparently this
was conjectured earlier by Zauner in a Ph.D. thesis not available in
english.  See reference \cite{Zauner1999Quantum} and \url{http://www.imaph.tu-bs.de/qi/problems/23.html}.} that
the set
$\{U_{(p,q)}\phi\}$ for some fiducial $\phi\in\mathcal H$ forms a \emph{symmetric informationally complete positive operator valued measure} (SIC-POVM).  The defining condition of a SIC-POVM is a set of $d^2$ vectors $\{\phi_k\}$ such that
\begin{equation}
|\ip{\phi_k}{\phi_j}|^2=\frac{\delta_{kj}d+1}{d+1}.\label{SICPOVMdef}
\end{equation}
The set is called symmetric since the vectors have equal overlap.  The POVM is formed by taking the projectors onto the one-dimensional subspaces spanned by the vectors.  It is informationally complete since these $d^2$ projectors span the space of linear operators acting on $\mathcal H$.

As of writing, it is an open question whether SIC-POVMs exist in
every dimension. Although numerical evidence suggests this to be the case \cite{Scott2010Symmetric}.

For the remainder of this section we assume, for any dimension $d$, a SIC-POVM exists.  Define the operators $P_k:=\frac1d\phi_k\phi_k^*$.  Then, define the \emph{symmetric-representation} of a quantum state $\rho$ as
\begin{equation}\label{SIC-representation}
\mu_\rho(k)=\Tr(\rho P_k).
\end{equation}
This is a probability distribution and in particular it is \emph{the} probability distribution for the POVM measurement formed by the effects $\{P_k\}$.  As we have noted, this is an informationally complete measurement. Therefore, the density matrix can be reconstructed from the probabilities via
\begin{equation}\label{SIC reconstruction}
\rho=d(d+1)\sum_{k=0}^{d^2-1}\mu_\rho(k)P_k-\id.
\end{equation}

When viewed as a mapping, this representation is a bijection from the convex set of density matrices to a convex subset of the $d^2$-dimensional probability simplex.

\subsubsection{Application: Quantum Bayesianism}

Quantum Bayesianism \cite{Schack2001Quantum,Caves2002Quantum,Fuchs2002Quantum,Fuchs2009QuantumBayesian,Fuchs2010QBism} is an interpretation of quantum theory which sheds new light on not only the tradition ``foundational'' problems (the ``measurement problem'', for example) but also many concepts in quantum theory, such as the ``unknown quantum state''\cite{Caves2002Unknown}.  A key realization is the mathematical and conceptual sufficiency of viewing quantum states as the probability distribution via the Born rule for a fixed POVM $\{E_k\}$.  The only remaining freedom is which one.

One ideal is to have the POVM elements orthogonal: $\Tr(P_k P_j)=\delta_{kj}$.  The statement that is not possible is equivalent to theorem \ref{theorem:noclassical}.  Next, then, we desire them to be as close to orthogonal as possible.  Formally, we want to minimize the quantity
\[
F = \sum_{kj} (\Tr(P_kP_j)-\delta_{kj})^2.
\]
This expression is minimized if and only if the $\{P_k\}$ form a SIC-POVM.  Using the reconstruction formula in equation \eqref{SIC reconstruction} it can be shown that, in terms of this SIC-representation, the Born rule for a measurement $\{E_j\}$ given state $\mu(k)$ is
\begin{equation}\label{SIC born}
\Pr(\textrm{outcome }j)=\sum_k \left(\mu(k)-\frac1d\right)\xi(j|k),
\end{equation}
where $\xi(j|k)=\Tr(E_jP_k)$.  In the same sense as the SIC-POVM being as close to orthogonal as possible, equation \eqref{SIC born} is as close as possible to the classical Law of Total Probability.  Effort is being made to use equation \eqref{SIC born} as a starting point for a natural set of axioms which would single out quantum theory.

\section{Unification of the quasi-probability representations via frames\label{S:Frame}}

\subsection{Introduction to frames\label{S:frames}}

A \emph{frame} can be thought of as a generalization of an orthonormal basis \cite{Christensen2003An}.  However, the particular Hilbert space under consideration here is not $\mathcal H$.  Considered here is a generalization of a basis for $\herm$, which is the set of Hermitian operators on an complex Hilbert space of dimension $d$.  With the trace inner product $\ip{  A}{  B}:=\Tr(  A  B)$, $\herm$ forms a \emph{real} Hilbert space itself of dimension $d^2$.  Let $\Lambda$ be some set of cardinality $d^2 \leq \abs{\Lambda} <\infty$.

A frame\footnote{Frames have been considered in the context of quantum theory for other purposes in \cite{Scott2006Tight,Bisio2009Optimal}.} for $\herm$ is a set of operators $\mathcal F:=\{ F(\lambda)\}\subset\herm$ which satisfies
\begin{equation}\label{def_discrete_frame}
a\norm{A}^2\leq\sum_{\lambda\in\Lambda} \Tr[ F(\lambda){A}]^2\leq b\norm{A}^2,
\end{equation}
for all $A\in\herm$ and some constants $a,b>0$.  This definition generalizes a defining condition for an orthogonal basis $\{ B_k\}_{k=1}^{d^2}$
\begin{equation}\label{def_basis}
\sum_{k=1}^{d^2}\Tr[{B_k}{A}]^2 = \norm{A}^2,
\end{equation}
for all $A\in\herm$.  The mapping $A\mapsto\Tr[ {F(\lambda)}{A}]$ is called a \emph{frame representation} of $\herm$.

A frame $\mathcal D:=\{ D(\lambda)\}$ which satisfies
\begin{equation}\label{def_dual}
A=\sum_{\lambda\in\Lambda} \Tr[{ F(\lambda)}{A}] D(\lambda),
\end{equation}
for all $A\in\herm$, is a \emph{dual frame} (to $\mathcal F$).  The \emph{frame operator} associated with the frame $\mathcal F$ is defined as
\begin{equation*}\label{def_frameop}
 S(A):=\sum_{\lambda\in\Lambda}  \Tr[{ F(\lambda)}{A}]F(\lambda).
\end{equation*}
If the frame operator is proportional to the identity \emph{superoperator}, $ S=a\tilde \id$, the frame is
called \emph{tight}. The frame operator is invertible and thus every
operator has a representation
\begin{align}\label{frame_decomposition}
A= S^{-1} SA=\sum_{\lambda\in\Lambda} \Tr[{ F(\lambda)}{A}] S^{-1} F(\lambda).
\end{align}
The frame $ S^{-1}\mathcal F$ is called
the \emph{canonical dual frame}.  When $\abs {\Lambda}=d^2$, the canonical dual frame is the unique
dual, otherwise there are infinitely many choices for a dual frame.  A tight frame is ideal from the perspective that its canonical dual is proportional to the frame itself.  Hence, the reconstruction is given by the convenient formula
\begin{align*}
A= S^{-1} SA=\frac{1}{a}\sum_{\lambda\in\Lambda} \Tr[ {F(\lambda)}{A}] F(\lambda)
\end{align*}
which is to be compared with
\begin{align*}
A=\sum_{k=1}^{d^2}\Tr[{ B_k}{A}] B_k
\end{align*}
which defines $\{ B_k\}_{k=1}^{d^2}$ as an orthonormal basis.

\subsection{Unification via the necessity and sufficiency of frames\label{S:framenecessary}}

Recalling the formal definition (\ref{def:qp rep quant theory}) of a quasi-probability representation, we have the following theorem
\begin{theorem}\label{theorem:frame}
Two functions $\mu$ and $\xi$ constitute a quasi-probability representation if and only if
\begin{align*}
\mu_\rho(\lambda)&=\Tr[\rho F(\lambda)]\\
\xi_E(\lambda)&=\Tr[E D(\lambda)],
\end{align*}
where $\{F(\lambda)\}$ is a frame and $\{D(\lambda)\}$ is one if its duals.
\end{theorem}
This was proven in reference \cite{Ferrie2009Framed}\footnote{Compare this to a similar result in a more operational setting in references  \cite{PortaMana2004Probability,PortaMana2004Why}.}.  This theorem allows us to make the following statement which is equivalent to the no-classical-representation theorem (\ref{theorem:noclassical}) and negativity theorem (\ref{theorem:neg}): there does not exists two frames of positive operators which are dual to each other.

With this results we can create quasi-probability representations of the whole operational formalism of quantum theory, not just states.  First, we chose one of the discrete quasi-probability functions described in section \ref{S:Finite}. Second, we identify the frame which gives rise to it.  Lastly, we compute its dual frame to obtain the part of the quasi-probability representation mapping, $\xi$, which takes measurements to functions on the space $\Lambda$.

Suppose instead the functions $\mu$ and $\xi$ are defined via
\begin{align*}
\mu_\rho(\lambda)&=\Tr[\rho F(\lambda)]\\
\xi_E(\lambda)&=\Tr[E F(\lambda)],
\end{align*}
where $\{F(\lambda)\}$ is a frame.  Then,
\begin{enumerate}[(a)]
\item $\mu_\rho(\lambda)\in[0,1]$ and $\sum_\lambda \mu_\rho(\lambda)=1$,
\item $\xi_E(\lambda)\in[0,1]$ and $\xi_\id(\lambda)=1$,
\item $\Tr(\rho E)=\sum_{\lambda,\lambda'}  \mu_\rho(\lambda) \xi_E(\lambda') \Tr[D(\lambda)D(\lambda')]$.
\end{enumerate}
In reference \cite{Ferrie2008Frame} this representation was called a \emph{deformed} probability representation since states and measurements are represented as true probabilities but the law of total probability is deformed.

The frame formalism also provides a convenient transformation matrix to map between representations.  We have
\begin{align*}
\mu_\rho(\lambda)&=\Tr[\rho F(\lambda)]\\
&=\sum_{\lambda'} \Tr[\rho F'(\lambda')]\Tr[D'(\lambda')F(\lambda)]\\
&=\sum_{\lambda'} T_{\lambda'\lambda} \mu'(\lambda'),
\end{align*}
where the matrix $T_{\lambda'\lambda}$ is the symmetric matrix which takes the $\mu$ representation to $\mu'$ representation.

\subsection{To infinity and beyond\label{S:infinite generalization}}

Given a quasi-probability representation note that the frame satisfies
\[
\sum_{\lambda\in\Lambda} F(\lambda) =\id.
\]
Thus, if the quasi-probability representation satisfies $0 \leq \mu(\lambda)\leq 1$, the frame is an informationally complete positive operator valued measure (IC-POVM)\footnote{Frames have also been used in definition of informationally completeness in the context of tomography in reference \cite{Bisio2009Optimal}.}.  Similarly, the dual frame satisfies
\[
\Tr[D(\lambda)]=1,
\]
for all $\lambda\in\Lambda$.  Thus, if the the quasi-probability representation satisfies $0\leq\xi(\lambda)\leq1$, the dual frame is a set of density operators.  The definitions and results we have so far considered are tailored to the case $d<\infty$ -- that is, finite dimensional quantum theory.  Now we will extend them to infinite dimensions as done in reference \cite{Ferrie2010Necessity}.

We allow for now the dimension of the Hilbert space $\mathcal H$ to be arbitrary and let $\meas$ be a measurable space, where $\Sigma$ is a $\sigma$-algebra. Over this space, $\smeas$ denotes the bounded signed measures while $\measf$ denotes the bounded measurable functions.  A signed measure generalizes the usual notion of measure to allow for negative values.  The classical states are the probability measures $\cstate\subset\smeas$.  We define the generalization of frame to the space of trace-class operators $\tsa$.

We generalize the definition of informationally complete observable for infinite dimensions and define an \emph{operator valued measure} as a map $F:\Sigma\to\bsa$ satisfying $F(\emptyset)=0$, $F(\Omega)=1$ and
\[
F\left(\bigcup_{i=1}^\infty B_i\right)=\sum_{i=1}^\infty B_i,
\]
where the sets $B_i\in\Sigma$ are mutually disjoint and the sum converges in the weak sense.  A \emph{frame} for $\tsa$ is an operator valued measure $F$ for which the map $T:\tsa\to\smeas$,
\[
T(W)(B):=\Tr(WF(B)),
\]
is injective.  The map $T$ is called a \emph{frame representation} of $\tsa$.

Similarly, generalizing the reconstruction formula defining the dual frame for finite dimensions yields the generalized notion of a dual.  That is, given a frame $F$, a \emph{dual frame} to $F$ is a map $D:\Omega\to\bsadual$, the dual space, for which the function
\[
(SA)(\omega):=(D(\omega))(A)
\]
is measurable and satisfies
\begin{equation}\label{reconstruction formula}
A=\int_\Omega SA dF,
\end{equation}
for all $A\in\bsa$.

On the other hand, we define a quasi-probability representation of quantum mechanics is a pair of mappings $T:\qstate\to\smeas$ and $S:\qeff\to\measf$ such that
\begin{enumerate}
\item $T$ and $S$ are affine, $T$ is bounded.
\item $T\rho(\Omega)=1$.
\item $S(0)=0$.
\item For all $\rho\in\qstate$ and $E\in\qeff$,
\begin{equation}\label{LTP in classical rep2}
\Tr(\rho E)=\int_\Omega  (SE) d(T\rho).
\end{equation}
\end{enumerate}
In reference \cite{Ferrie2010Necessity} it was shown that for these generalized definitions frame representations and quasi-probability representations remain equivalent and the following ``negativity theorem'' was proven:
\begin{theorem}
A quasi-probability representation of quantum mechanics must have, for some $\rho\in\qstate$, $E\in\qeff$ either $(T\rho)(B)<0$ for some $B\in\Xi$ or $(SE)(\omega) \not\in [0,1]$ for some $\omega\in\Omega$.
\end{theorem}

Being more general, this theorem implies the previous three equivalent ``negativity theorems'' and says essentially the same; a classical representation does not exist and within a quasi-probability representation negativity must appear in the representation of the states or measurements or both.

\section{Negativity in quasi-probability representation\label{S:negativity}}
There have been a variety of approaches to the problem of characterizing what is non-classical about quantum theory. In the previous section one such notion of non-classicality was considered: the requirement of ``negative probability'', or simply \emph{negativity}.  However, the negativity theorem leaves open the question of the interpretation of negativity in any \emph{particular} representation.  In section \ref{S:negativity_ss} below we discuss the ways in which negativity can be applied as a criterion for quantumness with respect to particular choices of representation.   In section \ref{S:context locality} the traditional ideas of \emph{contextuality} and \emph{nonlocality} in relation to negativity \footnote{These are diverse and rich fields of study in their own right.  A starting point for the interested reader on contextuality is reference \cite{Held2008KochenSpecker} and quantum nonlocality is reference \cite{Redhead1989Incompleteness}.}.

\subsection{Negativity as an indicator of quantumness\label{S:negativity_ss}}

There is a strong tradition in physics of considering negativity of the Wigner function as an indicator of non-classical features of quantum states.  The non-classical features attributed to negativity of \emph{the original} Wigner function include quantum nonlocality \cite{Bell2004Speakable,Kalev2009Inadequacy}, quantum chaos \cite{Paz1993Reduction} and quantum coherence \cite{Paz1993Reduction}.  In quantum optics, however, tradition has been to use the Q- and P-functions (section \ref{S(infinite):other}) to define quantumness \cite{Mandel1986NonClassical}.  Recall equation (\ref{Glauber-Sudarshan function}), where the P-function of $\rho$ was defined implicitly through
\begin{equation*}
\rho=\int d^2\alpha P(\alpha) \op{\alpha}{\alpha}.
\end{equation*}
If $P$ has the properties of a probability distribution, then the state is a mixture of coherent states. Coherent states are minimum uncertainty states and this fact is often cited when it is stated that such a state is ``the most classical'' of the quantum states of light.  More specifically, if $P$ is a probability distribution then the quantum field cannot display genuine quantum optical effects and can be simulated by a stochastic classical electromagnetic field \cite{Walls1995Quantum}.  Technically, however, $P$ is not a function but a distribution which can be highly singular. Thus $P$ functions which are not classical distributions are difficult to experimentally prepare and verify; although, recent progress has been made \cite{Kiesel2008Experimental}.

Effort has been extended beyond qualitatively defining negativity as quantumness to \emph{quantifying} quantumness via negativity. In terms of the Wigner function, the \emph{volume} of the negative parts of the represented quantum state has been suggested as the appropriate measure of quantumness \cite{Kenfack2004Negativity}.  The \emph{distance} (in some some preferred norm on $\herm$) to the convex subset of positive Wigner functions was suggested to quantify quantumness in reference \cite{Mari2010Directly}.  This was also done in references \cite{Giraud2008Classicality,Giraud2010Quantifying} for a finite analogs of the P- and Q-functions rather than the Wigner function.

The main difficulty with interpreting negativity in a particular quasi-probability representation as a criterion for or definition of quantumness is the non-uniqueness of that particular quasi-probability representation.  We can always find a new representation in which any given state admits a non-negative quasi-probability representation.  Recall, in fact, that in some representations all states are non-negative.  Thus, negativity of some state $\rho$ in one particular arbitrary representation is a meaningless notion of quantumness \emph{per se}.

An alternative approach to establishing a connection between quantumness and negativity is to start by assuming some criterion for quantumness and then finding a choice of representation in which this criterion is expressed via negativity.  This approach has been applied in the context of multi-partite systems for which entanglement is presumed to provide a criterion for quantumness.  Entanglement is a kind of correlation between two quantum systems which cannot be achieved for classical variables and is one of the central ingredients in quantum information theory \footnote{For a recent review of entanglement, see \cite{Horodecki2009Quantum}.}.  Recall that a density matrix $\rho$ is \emph{entangled} if it \emph{cannot} be written as a convex combination of the form $\rho=\sum_k p_k  \rho^{(1)}_k\otimes \rho^{(2)}_k,$ for all $k$, where $ \rho^{(1)}_k$ and $\rho^{(2)}_k$ are states on the individual subsystems.  Consider a \emph{product-state frame} constructed out of frames for two subsystems.  That is, consider the frame $\{F^{(1)}(\lambda)\otimes F^{(2)}(\lambda')\}$, where each $\{F^{(j)}(\lambda)\}$ is a set of density matrices composing a frame.  Then, if we represent a quantum state using the dual frame, we a have a quasi-probability representation in which states with negativity are entangled.  Explicit constructions of such quasi-probability representations were developed by Schack and Caves \cite{Schack2000Explicit}.  An optimization procedure to find the representation with the \emph{minimum} amount of negativity was given in reference \cite{Sperling2009Representation}.

An obvious limitation with the above approach is that entanglement cannot capture any notion of quantumness for single quantum systems.  A second, more subtle, issue is that identification of entanglement as ``the'' crucial non-classical resources is problematic in certain branches of quantum information science.  The most striking example questioning the role of entanglement in quantum information theory is DQC1 (deterministic quantum computing with one clean qubit).  DQC1 \cite{Knill1998Power} is a model of computation which refers to any algorithm which satisfies the following (or a modification not requiring exponentially more resources)\footnote{This model of computation has served as the basis for various definitions of complexity classes, also called DQC1.  There are many open questions in this line of research and the interested reader should consult references \cite{Shor2008Estimating,Ambainis2006Computing,Shepherd2006Computation}.}:
\begin{enumerate}
\item its input consists of a single pure state in the first \emph{control} register and the remaining $n$ registers are in the maximally mixed state $\rho=2^{-n}\id$;
\item the input state is subjected to a unitary $U_n$ controlled by the state of the first register;
\item the output is a statistical estimate of $2^{-n}\Tr(U_n)$ (achieved by measuring the average of control bit in the $Z$ basis).
\end{enumerate}
DQC1 appears to be a non-trivial computational model which has been shown to have exponential advantages over (known) classical algorithms in the the follow areas: simulation of quantum systems \cite{Knill1998Power}, quadratically signed weight enumerators \cite{Knill2001Quantum}, evaluating the local density of states \cite{Emerson2004Estimation}, estimating the average fidelity decay under quantum maps \cite{Poulin2004Exponential} and estimating the value of Jones polynomials \cite{Shor2008Estimating}.

In the DQC1 model, the bipartite split between the control qubit and the rest contains no entanglement and in reference \cite{Datta2005Entanglement} it was shown that there is a vanishingly small amount of entanglement across any other bipartite splitting.  This suggests it is unlikely that entanglement is responsible for the speed-up provided by DQC1 \cite{Datta2008Studies}.  Conceptually, computation is a local task with complex dynamics and may not require the non-local, Bell-inequality-violating correlations of entanglement \cite{Laflamme2002NMR}.  A sentiment issued in reference \cite{Laflamme2002NMR} and reiterated recently by Vedral \cite{Vedral2010The} is that no one single criteria can capture quantumness and perhaps even the resources necessary for the quantum advantage must be studied on a case-by-case basis.

An important consideration for all of the above approaches is that the notion of a quantum state, considered in isolation, is operationally meaningless. Comparison with experiment always requires specifying both a state (a preparation procedure) and a measurement.  Consider two experiments, one which prepares a product state and measures the state by projecting onto an entangled basis, and a second which prepares an entangled state and measures that state in a product basis. Both experiments produce the same statistical predictions, but only the second is considered non-classical when considering the state in isolation.  As emphasized in references \cite{Spekkens2008Negativity,Ferrie2008Frame} we can overcome this obvious deficiency if we consider the whole operational set-up -- states \emph{and} measurements.  In this way, the existence of a positive quasi-probability representation implies the existence of a non-contextual ontological model and vice versa.

The formalism of references \cite{Ferrie2008Frame,Ferrie2009Framed,Ferrie2010Necessity} shows the necessity of negativity when considering a representation of the full quantum formalism.  That is, the negativity theorem (theorem \ref{theorem:neg}) applies to quasi-probability representations of quantum theory as a whole. However, the negativity theorem may not apply if we consider a specific experiment, device, or protocol which may not faithfully reproduce the full power of quantum theory.  This work suggests the following promising approach: define a classical representation \emph{of an experiment} as the existence of a frame and dual for which the convex hull of the experimentally accessible states and measurements have positive representation.  Then, we can conclude that negativity, taken to mean the absence of any representation satisfying the above conditions, corresponds to quantumness.

The above criterion for classicality was considered in reference \cite{Schack1999Classical} to question the quantum nature of proposed NMR quantum computers.  However, as noted there, the immediate objection is the following: the states and measurements can be represented by classical probabilities while the transformation between them may not be represented by classical stochastic maps.  That is, a truly classical model must represent each applied transformation in a experiment as classical stochastic mapping.  In reference \cite{Schack1999Classical}, such stochastic maps were identified for the set of NMR experiments reported at the time \footnote{Note, however, that the reasonable requirement of an \emph{efficient} classical model was not met.}.

The scope of quantum theory that has been consider thus far can be thought of as \emph{kinematical}; only the description of experimental configurations is of concern. The traditional approach to quantum theory (quantum \emph{mechanics}) focuses on how and why quantum systems change in time. Using the Wigner function formalism to describe the dynamical transformations predicted by quantum mechanics yields the dynamical law
\begin{equation}\label{Moyal_bracket}
\frac{\partial \mu_\rho}{\partial t} =\{H,\mu_\rho\}+ \sum_{n=1}^\infty \frac{1}{2^{2n}(2n+1)!} \frac{\partial^{2n+1} H}{\partial q^{2n+1}} \frac{\partial^{2n+1} \mu_\rho}{\partial p^{2n+1}},
\end{equation}
where $\mu_\rho$ is the Wigner function and $H$ is the classical Hamiltonian and $\{H,\mu_\rho\}$ is the classical Poisson bracket.  Notice then that Equation \eqref{Moyal_bracket} is of the form ``classical evolution'' + ``quantum correction terms''.  Using this formalism, one can then do more than discuss which experimental procedures are classical.  Now one can discuss the \emph{transitions} between quantum and classical descriptions, a process known as \emph{decoherence} \cite{Paz1993Reduction,Habib1998Decoherence,Joos2003Decoherence}.  The representation of the dynamics was also studied for the spherical phase space in  \cite{Zueco2007Bopp,Kalmykov2008Phase}.  The goal is to compare the representation of the quantum dynamics (be it the Schrodinger equation or the more general \emph{master equation}) to the natural classical dynamics of the representation's phase space.  The challenge for finite dimensional systems is that no natural notion of discrete phase space exists for classical system.  This problem has been recently studied by Livine \cite{Livine2010Notes} by introducing a \emph{discrete differential calculus} for the discrete phase spaces of Wootters.  However, beyond these few examples, transformations and dynamics have not been studied anywhere near to the extent that states have for quasi-probability representations and presents itself as a open problem.

\subsection{Traditional contextuality and nonlocality\label{S:context locality}}

The traditional definition of contextuality evolved from a theorem
which appears in a paper by Kochen and Specker \cite{Kochen1967Problem}.  The
Kochen-Specker theorem concerns the standard quantum formalism:
physical systems are assigned states in a complex Hilbert space
$\mathcal H$ and measurements are made of observables represented by
Hermitian operators.  The theorem establishes a contradiction
between a set of plausible assumptions which together imply that
quantum systems possess a consistent set of pre-measurement values
for observable quantities. Let $\mathcal H$ be the Hilbert space
associated with a quantum system and $ A\in\herm$ be the operator
associated with some observable.  The function $f_\psi(A)$
represents the value of the observable when the system is in
state $\psi$.  One assumption used to derive the contradiction is
that for any function $F$, $f_\psi(F(A))=F(f_\psi(A))$.  This is plausible because, for
example, we would expect that the value of $A^2$ could be obtained
in this way from the value of $A$.

The conclusion of Kochen-Specker theorem implies the
following counterintuitive example \cite{Isham1995Lectures}.  Suppose three
operators $ A$, $ B$, and $  C$ satisfy $[ A, B]=0=[ A,  C]$, but $[
B,  C]\neq0$.  Then, the value of the observable $A$ will depend on
whether observable $B$ or $C$ is chosen to be measured as well. That
is, assuming that physical systems do possess values which can be
revealed via measurements, the value of $A$ depends on the \emph{context} of the
measurement.

What the Kochen-Specker theorem establishes then is the mathematical
framework of quantum theory does not allow for a
\emph{non-contextual} model for pre-measurement values.  This fact
is often expressed via the phrase ``quantum theory is contextual''.

The original notion of contextuality in lacking in the sense that it only applies to the standard formalism of quantum theory and does not apply to general operational models.  This problem was addressed by Spekkens as discussed in section \ref{sec:probtoquasi} above.  The notion of contextuality defined by Spekkens is more general; one can recover the original assumptions of Kochen-Specker by assuming that the projector valued measures in the spectral resolutions of observables are represented by dispersion free (0-1 valued) conditional probabilities (these are also called \emph{sharp indicator functions}) \footnote{Cabello has also generalized the notion of contextuality to POVMs \cite{Cabello2003KochenSpecker}.  Again, however, the additional assumption of dispersion free condition probabilities is used.  See \cite{Grudka2008There} for an elaboration on this point.  For a more broad discussion on contextuality see \cite{Morris2009Topics} and \cite{Harrigan2007Ontological}.}.  Since the set of fewer assumptions already contains a contradiction when taken in conjunction, the addition of the assumption of Kochen-Specker is unnecessary.  Thus, we need only consider the more general notion of contextuality we have already defined.  This more general notion of contextuality has also recently been subject to experimental tests \cite{Spekkens2009Preparation}.

A hidden variable theory originally formulated by de Broglie and later by Bohm \cite{Bohm1989Quantum} is perhaps the most famous example of an ontological model of quantum theory.  The model assumes that for a given experimental configuration, there exists a particle with well defined trajectory and a quantum state $\psi$.  The hidden variable is the position of the particle in real space.  That is, the classical state space is $\Lambda=\mathbb R^3\times\mathcal H$.  The Hilbert space is included in the state space as its serves as a wave which guides the particle.  If at any time the particle is distributed according to quantum probability distribution $\abs{\psi}^2$, it remains so.  Thus, so long as it is assumed that the particle is prepared according to this distribution, the model provides the same predictions as the standard formulation of quantum theory.

Note that this model does not fit into the framework of quasi-probability representations.  Exactly as it was for the Beltrametti-Bugajski model, the de Broglie-Bohm model does not consider the entire range of possible quantum states.  Where a classical representation contains a convex-linear mapping $\rho\mapsto\mu(\lambda)$ for all $\rho\in\den$, the de Broglie-Bohm model considers only a mapping with domain $\mathcal H$.  Bell notes that \cite{Bell2004Speakable} ``in the de Broglie-Bohm theory a fundamental significance is given to the wavefunction, and it cannot be transferred to the density matrix.''

Bell \emph{does not} claim that the situation is such that the de Broglie-Bohm model \emph{cannot} be extended to include density operators.  The key words in his comment are ``fundament significance''.  Indeed, the de Broglie-Bohm model \emph{can} be extended to include density operators provided this extension is either contextual or contains negativity.  In either case, the pure states (wavefunctions) retain their significance while the density operators possess non-classical features.  As an example, the de Broglie-Bohm model could be such that $(\rho, c_\mathcal P)\mapsto\mu_{c_\mathcal P}(\lambda)$ where each preparation consists of a density operator $\rho$ supplied with a context $c_\mathcal P$ which specifies a particular convex decomposition of $\rho$ into pure states.  Such a model would be preparation contextual.

The non-locality debate was initiated by a paper by Einstein, Podolsky and Rosen (EPR) \cite{Einstein1935Can} where it was argued that quantum mechanics is \emph{incomplete} (each element of physical reality does not have a counterpart in quantum theory) if special relativity remains valid.  The latter means physical causation must be local or events cannot have causes outside of their past light cones.  Using a particular spatially separated quantum system, and some standard quantum theory, EPR concluded that quantum mechanics is either incomplete or nonlocal (or both!).  Locality was such a desired property of any theory that quantum mechanics was concluded to be incomplete.  That is, there must be elements of physical reality (hidden variables) which quantum mechanics does not account for.

The argument of EPR was reformulated by Bohm \cite{Bohm1989Quantum} for two qubits.  The argument is built around the following hypothetical experiment.  Two parties, Alice and Bob, are at distant locations with a source midway between them creating quantum systems described by the quantum state
\begin{equation}\label{bellstate}
\psi=\frac{1}{\sqrt{2}}(\phi_1\otimes\phi_2 - \phi_2\otimes\phi_1),
\end{equation}
where $\{\phi_1,\phi_2\}$ is an orthonormal basis for a qubit.  One particle is sent to Alice and the other to Bob.  Alice performs the projective two-outcome measurement $\{P_1,P_2\}$ on the particle which was sent to her.  The state in equation \eqref{bellstate} is such that Alice, once she performs her measurement, she can predict with certainty the outcome Bob receives when he performs the same measurement at his side of the experiment \emph{regardless of whether or not the measurement events are spacelike separated (i.e. nonlocal)}.  For example, Alice could perform the measurement $\{\phi_1\phi_1^\ast,\phi_2\phi_2^\ast\}$.  According to the collapse postulate, if Alice registers the first outcome, Bob particle will immediately collapse to $\phi_2$ and he is certain to obtain the second outcome if he were to make the same measurement.  Therefore, unless there exists hidden variables which pre-determine the possible outcomes when the particles are created, quantum theory is \emph{nonlocal}.

Bell later investigated the possibility of finding the hidden variables Einstein thought to exist \cite{Bell2004Speakable}.  He noted immediately that the de Broglie-Bohm theory was such a theory yet in contained an astonishingly nonlocal character.  He was able to prove that any hidden variable theory of quantum phenomena must possess nonlocal features.  This is now called Bell's theorem.

The proof is by contraction and follows the general line of reasoning which lead to the negativity theorem: build a mathematical model with assumptions that can be identified with (or motivated by) some notion of classicality then prove that quantum theory does not satisfy these assumptions.  Consider the EPR experimental setup.  Alice and Bob can each perform a two-outcome measurement with outcomes labeled $A$ and $B$, respectively.  Without loss of generality, the outcomes can be assigned numerical values $A,B=\pm 1$.

Suppose there exist a classical state space $\Lambda$ (i.e. a set of hidden variables or, as we have called it above, an ontology) which serves to determine the outcomes $A$ and $B$.  Probabilistic knowledge of the state is represented by a density $\mu(\lambda)\geq0$ which is normalized
\[
\int_\Lambda d\lambda \mu(\lambda)=1.
\]
The different measurements Alice and Bob can perform are parameterized by detector settings $a$ and $b$, respectively.  Locality is enforced by assuming that the outcomes $A$ and $B$ depend only the local detector settings and the global state.  That is $A=A(a,\lambda)$ is allowed but $A=A(a,b,\lambda)$ is not.  Define the correlation function
\begin{equation}\label{bell_correlation}
C(a,b)=\int_\Lambda d\lambda A(a,\lambda)B(b,\lambda)\mu(\lambda).
\end{equation}
Bell's theorem states that the correlations obtained in the EPR experiment (i.e. a particular quantum experiment) cannot satisfy this equation.  The proof follows by deriving an inequality from equation \eqref{bell_correlation} such as
\begin{equation}\label{CHSH}
\abs{C(a,b)-C(a,c)}\leq 1+ C(b,c).
\end{equation}
This inequality holds for any hidden variable model which satisfies the locality assumption.  For the quantum state in equation \eqref{bellstate}, the inequality is violated.  This is the contradiction between the quantum theory and a local hidden variable model which proves Bell's theorem.

It was noted that the assumptions which go into the hidden variable models first considered by Bell imply those models are \emph{deterministic}.  That is, the theorem did not exclude models which suggested quantum theory only provides \emph{stochastic} (or probabilistic) information of the possible outcomes of measurements.  Bell later extended the theorem to include such models.  For the EPR experimental setup, let the conditional probability of outcome $A=1$, for Alice, given the state (hidden variable) is $\lambda\in\Lambda$ be denoted $M_A(\lambda)$ and similarly define $M_B(\lambda)$ for Bob.  Now denote the conditional joint probability of the simultaneous outcomes $A,B=1$ by $M_{AB}(\lambda)$.  Fine \cite{Fine1982Some} defines a \emph{stochastic hidden variable model} as one which satisfies
\begin{equation}\label{stocFine1}
\Pr(A=1)=\int_\Lambda d\lambda M_A(\lambda)\mu(\lambda)
\end{equation}
and
\begin{equation}\label{stocFine2}
\Pr(A=1,B=1)=\int_\Lambda d\lambda M_{AB}(\lambda)\mu(\lambda).
\end{equation}
If $M_{AB}(\lambda)=M_A(\lambda)M_B(\lambda)$, then the model is \emph{factorizable}.  Bell claimed this also encoded the assumption of locality.  Again, it can be shown that quantum theory is in contradiction with an inequality derived from these assumptions.  Fine showed that a factorizable stochastic hidden variable model exists for the EPR-type correlation experiment if and only if a deterministic hidden variable model exists for the experiment.  Since the latter is ruled out, the former is also ruled out.

It is often stated that \emph{the} consequence of Bell's theorem is ``quantum theory is non-local''.  However the theorem only states that quantum theory does not satisfy the assumptions which go into a classical model \emph{Bell defines as local}.  It is not necessary that the locality assumption is violated.  Nor is it unanimously agreed that the mathematical condition Bell refers to as locality in the stochastic hidden variable models reflects any physical significance \cite{Fine1982Some}.  Next, we will see how such a claim can be supported by appealing to the notions of negativity.

Notice that a (non-factorizable) stochastic hidden variable model is exactly a classical representation. Then factorizability (Bell locality) is an additional requirement.  However, the negativity theorem rules out a classical representation independently of any assumption of locality for the most general quantum experiment.  The negativity theorem implies that one (or more) of the constraints which go into the definition of a classical representation are false.  Indeed, relaxing the assumptions of positive probability yields a quasi-probability representation which is not in conflict with quantum theory.  The existence of \emph{positive} probability distributions in the hidden variables Bell attempts to rule out is often considered an unquestionable assumption.  However, if it is the case that negative probability encodes something about Nature that is independent of locality, then  it is not necessarily the case that Bell’s theorem implies
nonlocality
 \cite{Han1996Explicit,Cereceda2000Local,Rothman2001Hidden}.

\section{Summary\label{S:Conclusion}}
In this article, we have seen how quasi-probability representations arise naturally out of the search for an ontological model of quantum theory.  A non-contextual ontological model implies the existence of a classical representation, which was shown to be in contradiction with the quantum framework.  However, by relaxing the requirement of positive probabilities we have seen that many \emph{quasi}-probability representations are consistent with quantum theory, and indeed useful for quantum information theory.

We have shown how the theory of frames unifies the quasi-probability representations into a single formalism.  This frame formalism allows two canonical methods for incorporating measurements into the framework.  In the first, we use a \emph{dual} frame to represent measurements.  This gives rise to framework respecting the usual dot product between the represented functions (which is the classical rule of total probability) and is equivalent to the formal definition of a quasi-probability representation.  In the second method, the \emph{same} frame is used to represent states and measurements.  This gives rise to the \emph{deformed} probability representation.  Here, negativity need not appear in the representations of states and measurements but the Born rule is replaced by a deformation of the classical rule of total probability.

Finally, we have noted that there is significant interest surrounding the presence and interpretation of negativity in the quasi-probability representations.  In particular, the use of negative probability allows us to avoid the contradictions which had lead to conclusions that quantum theory is contextual (via the Kochen-Specker theorem) and nonlocal (via Bell's theorem).

We have discussed how negativity of \emph{states alone} in some particular representation is a meaningless criterion for quantumness \emph{per se}.  Alternatively, we can accept some well defined notion of a truly quantum state (such as an entangled state) and show that negativity is necessary for such state in some representation.  If we are ultimately to make use of our notion of quantumness, however, it needs to be defined for an operational task and hence incorporate the full operational formalism of quantum theory.  That is, it must incorporate measurements (which we have done here) and \emph{dynamical transformation} (which we have identified as a largely unexplored area of research).

\begin{acknowledgements}
I thank Joseph Emerson, Chris Fuchs, Ryan Morris, Cozmin Ududec, Rob Spekkens, Victor Veitch and Joel Wallman for many helpful discussions.  In particular, I thank Joseph Emerson for a careful reading of an early draft and suggestions for its improvement.   This work was financially supported by the Canadian government through NSERC.
\end{acknowledgements}

\appendix

\section{Notations and conventions\label{appendix:Mathematical primer}}
Quantum theory makes use of complex Hilbert spaces.  If these spaces are finite dimensional, then they are equivalent to inner product spaces.  Unless otherwise noted, a Hilbert space $\mathcal H$ will be assumed to have dimension $d<\infty$ and (for $\psi,\phi\in\mathcal H$) its inner product will be denoted $\ip{\psi}{\phi}$.  The following list defines some specials sets of linear operators acting on $\mathcal H$.

\begin{enumerate}
\item An operator $U$ satisfying
\[
U^\dag U= U U^\dag = \id,
\]
is called \emph{unitary}.  The set of all unitary operators is denoted $\mathbb U$.
\item An operator $A$ satisfying
\[
\ip{A\psi}{\phi}=\ip{\psi}{A\phi},
\]
for all $\psi,\phi\in\mathcal H$, is called \emph{Hermitian}.  The set of all Hermitian operators is denoted $\herm$.
\item An operator $E$ satisfying
\[
0\leq\ip{\psi}{E\psi}\leq1,
\]
for all $\psi\in\mathcal H$, is called an \emph{effect}.  The set of all effects is denoted $\eff$.
\item An effect $\rho$ satisfying
\[
\Tr(\rho)=1
\]
is called a \emph{density operator}.  The set of all density operators is denoted $\den$.
\item A set of effects $\{E_k\}$ which satisfy
\[
\sum_k E_k=\id
\]
is called a POVM (positive operator valued measure).  The set of all POVMs is denoted $\povm$.
\item An effect $P$ satisfying
\[
P^2=P
\]
is called a \emph{projector}.  The set of all projectors is denoted $\proj$.
\item A set of projectors $\{P_k\}$, each of rank 1, which satisfy
\[
\sum_k P_k=\id
\]
is called a PVM (projector valued measure).  The set of all PVMs is denoted $\pvm$.
\end{enumerate}

Note that from these definitions we have $\proj\subset\den\subset\eff\subset\herm$ and $\pvm\subset\povm$.  The set $\herm$ defines its own Hilbert space with inner product (for $A,B\in\herm$) $\ip{A}{B}:=\Tr(AB)$.  The dimension of $\herm$ is $d^2$.

 The following are some examples of important operators used in this paper.  Consider the operator $Z$ whose spectrum is $\spec(Z)=\{\omega^k: k\in \mathbb Z_d\}$, where $\omega^k=\w k$.  The
eigenvectors form a basis for $\mathcal H$ and are denoted
$\{\phi_k\}$. Consider also the operator defined by
$ X\phi_k=\phi_{k+1}$, where all arithmetic is modulo $d$.  Define $Y$ implicitly through $[ X, Z]=2i Y$.  The
operators $ Z,  X$ and $ Y$ are often called \emph{generalized
Pauli operators} since they are indeed the usual Pauli operators
when $d=2$.  The \emph{parity operator} is defined by $ P\phi_k=\phi_{-k}$.

\end{document}